\documentclass[]{interact}

\usepackage{epstopdf}
\usepackage{soul}
\usepackage[caption=false]{subfig}

\usepackage{bm}
\usepackage{multicol}
\usepackage{lipsum}
\usepackage{mwe}
\usepackage{graphicx}
\usepackage{subfig}
\usepackage{amssymb}
\newtheorem{assumption}{Assumption}
\usepackage{xcolor}
\usepackage{mathtools}
\usepackage{tabularx}
\usepackage{booktabs}
\usepackage[labelfont=bf,format=plain,justification=raggedright,singlelinecheck=false]{caption}
\usepackage{multirow, array, tabularx, graphicx}
\usepackage{booktabs} 
\usepackage{siunitx}  
\usepackage{adjustbox}

\usepackage{tikz}
\usetikzlibrary{decorations.pathreplacing}
\usetikzlibrary{shapes.misc, positioning}
\usepackage[normalem]{ulem}
\usepackage{tkz-euclide}
\usepackage{pgfmath}
\usepackage{pgf}
\usepackage{physics}
\usetikzlibrary{math} 
\usetikzlibrary{patterns}
\pgfdeclarepatternformonly{Rnd}
{\pgfpointorigin}{\pgfpoint{3cm}{3cm}}
{\pgfpoint{3cm}{3cm}}
{\foreach \i in {1,...,50} {%
    \pgfpathcircle{%
        \pgfpoint{rnd*.5cm}{rnd*.5cm}}{.25pt}}
    \pgfusepath{fill}
}

\usepackage{algorithmic}
\usepackage[ruled]{algorithm2e}
\usepackage{cases}

\usepackage{amsfonts}
\usepackage{graphicx} 
\usepackage{epstopdf}
\ifpdf
  \DeclareGraphicsExtensions{.eps,.pdf,.png,.jpg}
\else
  \DeclareGraphicsExtensions{.eps}
\fi

\usepackage[natbibapa,nodoi]{apacite}
\newcommand{\inP}{\mbox{$\,\stackrel{\scriptsize{\mbox{p}}}{\rightarrow}\,$}}
\newcommand{\as}{\mbox{$\,\,\stackrel{\scriptsize{\mathrm{a.s.}}}{\longrightarrow}\,\,$}}
\newcommand{\inD}{\mbox{$\,\stackrel{\scriptsize{\mbox{d}}}{\rightarrow}\,$}}

\theoremstyle{plain}
\newtheorem{theorem}{Theorem}[section]
\newtheorem{lemma}[theorem]{Lemma}
\newtheorem{corollary}[theorem]{Corollary}

\theoremstyle{definition}
\newtheorem{definition}[theorem]{Definition}

\theoremstyle{remark}
\newtheorem{remark}{Remark}

\begin{document}


\title{UNCERTAINTY QUANTIFICATION USING SIMULATION OUTPUT: BATCHING AS AN INFERENTIAL DEVICE
\author{
\name{Yongseok Jeon\textsuperscript{a}, Yi Chu\textsuperscript{b}, Raghu Pasupathy\textsuperscript{b}\thanks{CONTACT R. Pasupathy. Email: pasupath@purdue.edu} and Sara Shashaani\textsuperscript{a}}
\affil{\textsuperscript{a} Edward P. Fitts Department of Industrial and System Engineering, North Carolina State University, Raleigh, NC 27695, USA \\ \textsuperscript{b} Department of Statistics, Purdue University, West Lafayette, IN 47906, USA }
}
}

\maketitle

\begin{abstract}
We present \emph{batching} as an omnibus device for uncertainty quantification using simulation output. We consider the classical context of a simulationist performing uncertainty quantification on an estimator $\theta_n$ (of an unknown fixed quantity $\theta$) using only the output data $(Y_1,Y_2,\ldots,Y_n)$ gathered from a simulation. By \emph{uncertainty quantification}, we mean approximating the sampling distribution of the error $\theta_n-\theta$ toward: (A) estimating an ``assessment'' functional $\psi$, e.g., bias, variance, or quantile; or (B) constructing a $(1-\alpha)$-confidence region on $\theta$. We argue that batching is a remarkably simple and effective device for this purpose, and is especially suited for handling dependent output data such as what one frequently encounters in simulation contexts. We demonstrate that if the number of batches and the extent of their overlap are chosen appropriately, batching retains bootstrap's attractive theoretical properties of \emph{strong consistency} and \emph{higher-order accuracy}. For constructing confidence regions, we characterize two limiting distributions associated with a Studentized statistic. Our extensive numerical experience confirms theoretical insight, especially about the effects of batch size and batch overlap.
\end{abstract}

\begin{keywords}
output analysis, bias estimation, variance estimation, quantile estimation
\end{keywords}

\section{INTRODUCTION}\label{sec:intro} The context of this paper is a simulation experiment where data are collected in the service of estimating one or more fixed but unknown \emph{population parameters} associated with a simulated system. A typical example is a large hospital, a warehouse, or an airport, where a simulationist gathers data from computer experiments performed with the intent of estimating parameters such as the expected customer delay, a specific percentile of customer delay, the likelihood of delay exceeding a specified threshold, the expected length of ``long" delays, etc. Such \emph{point estimation}, however laborious, usually proceeds in a straightforward manner --- gather individual customer delay data from the simulation, and then construct the sample mean or sample percentile from the collected data as \emph{estimates} of the corresponding population parameters, in this case the population mean and the population percentile, respectively. Point estimation using simulation output is a basic topic that is usually extensively covered in most basic simulation textbooks such as~\cite{nel2013} and~\cite{2007law}.

\begin{remark}[Point Estimation] 
    The process of constructing an estimator $\theta_n$ of the population parameter $\theta$ is sometimes called \emph{point estimation} in classical statistics, see, for instance,~\cite[pp. 2]{1998lehcas}. This terminology is general and subsumes contexts where $\theta$ and $\theta_n$ reside in any appropriate space, e.g., Euclidean space, function space. We have, however, limited the use of such terminology to reduce the risk of a reader misunderstanding that $\theta$ and $\theta_n$ are restricted to being points in Euclidean space. Whenever convenient, we have thus preferred using the term \emph{estimation} instead of \emph{point estimation}, and \emph{estimator} instead of \emph{point estimator}. 
\end{remark}

This paper considers the question of assessing the error in a simulation point estimator. Specifically, any estimator constructed from output data generated by a stochastic simulation is \emph{random} in the sense that if the simulation experiment is repeated, the resulting output data, and the corresponding estimate of the same parameter, are bound to change. This leads the simulationist to ask if the extent of the fluctuations in the error can somehow be assessed without collecting any further data. More formally, can the uncertainty in the point estimator $\theta_n$ be quantified? Questions related to such error assessment are collectively and loosely called \emph{uncertainty quantification} here, and form the topic of this paper. Uncertainty quantification presents many open questions, especially when the point estimate is constructed from output data that form a time series.

\subsection{Uncertainty Quantification in the Simulation Context}\label{sec:statinf}
To make the question of uncertainty quantification using simulation output precise, suppose we have an output ``dataset'' $(Y_1,Y_2,\ldots,Y_n)$ of identically distributed $\mathcal{Y}$-valued random variables obtained somehow, e.g., using a simulation, in the service of estimating an unknown quantity $\theta \in \mathbb{R}^d$. For the purposes of this paper, we assume that $Y_j$ can be real ($\mathbb{R}$) or vector ($\mathbb{R}^d$) valued. (The results we present also extend to function-valued contexts, e.g., $Y_j$ represents the trace of the $j$-th customer or component, but we choose to remain in the Euclidean context to retain intuition.) Furthermore, $(Y_1,Y_2,\ldots,Y_n)$ can be data collected from a single long simulation run, or output data from $n$ independent replications. Formally, we assume that $(Y_1,Y_2,\ldots,Y_n)$ forms the initial segment of a \emph{stationary time-series} $\{Y_t, t \geq 1\}$. Especially important for the simulation context, $Y_1,Y_2, \ldots,$ may exhibit heavy serial correlation even though they each have the same distribution.  

Suppose that the output data $(Y_1,Y_2,\ldots,Y_n)$ are used appropriately to construct a point estimator $\theta_n$ of a fixed unknown parameter $\theta \in \mathbb{R}^d$. Again to anchor the reader's intuition, $\theta$ may represent a vector of quantiles of customer delays at various locations in a simulated network operating at steady state, in which case $\theta_n$ is the corresponding vector of estimators constructed from data obtained from a single long simulation run. Furthermore, while we refer to $\theta$ as a \emph{parameter}, the methods we present in this paper are better understood by viewing $\theta$ as a \emph{statistical functional} $\theta(P)$ of an unknown underlying probability measure $P$ that governs the simulation output data $(Y_1,Y_2,\ldots,Y_n)$. Accordingly, we use the notation $\theta$ and $\theta(P)$ interchangeably, and as appropriate; likewise, and as we shall see in more detail, we use the notation $\theta_n$ and $\theta(P_n)$ interchangeably and as appropriate, where $P_n$ is the empirical measure constructed from $(Y_1,Y_2, \ldots,Y_n)$. 

\begin{remark} (Simulation is Error Free) It is important that the question considered in this paper assumes that the model or simulation generating the output data $(Y_1,Y_2, \ldots,Y_n)$ is error free. This implies in turn that the population parameter $\theta$ associated with the distribution that generates the data $(Y_1,Y_2, \ldots,Y_n)$ is indeed the parameter that is of interest to the simulationist. Another implication is that questions such as \emph{model validation} that pertain to the correctness of the model do not lie within the purview of our treatment.
\end{remark}

For the purposes of this paper, \emph{uncertainty quantification} (UQ) in the simulation context refers approximating the \emph{sampling distribution} of the error\begin{equation}\label{esterr}\varepsilon_n := \theta_n - \theta,\end{equation} which is simply the distribution $P_{\varepsilon_n}$ of $\varepsilon_n$. Internalize the notion of the distribution of $\varepsilon_n$ by imagining the (impossible) experiment of gathering numerous realizations of $\varepsilon_n$, and then constructing a ``histogram.'' The key point of all uncertainty quantification is to approximate this sampling distribution of $\varepsilon_n$ without performing any additional simulation runs, that is, with only $(Y_1, Y_2, \ldots,Y_n)$. Approximating the sampling distribution of $\varepsilon_n$ is usually done in the service of two activities.
\begin{enumerate} \item[(A)] \emph{Estimating a vector of statistical functionals} $\psi: P_{\varepsilon_n} \to \mathbb{R}^d$, where $P_{\varepsilon_n}$ is the true unknown distribution of $\varepsilon_n$. Three typical examples of $\psi$ are the bias, variance, and quantiles associated with $\varepsilon_n$. The bias of $\theta_n$ is given by \begin{equation}\label{bias} \psi^{\text{b}}(P_{\varepsilon_n}) \equiv \mathbb{E}[\varepsilon_n] = \mathbb{E}\left[\theta_n\right] - \theta.\end{equation} The variance of $\varepsilon_n$ is given by \begin{equation} \label{var} \psi^{\text{v}}(P_{\varepsilon_n}) \equiv \Sigma_{\varepsilon_n} := \mathbb{E}\left[\left(\varepsilon_n - \mathbb{E}[\varepsilon_n]\right)\left(\varepsilon_n - \mathbb{E}[\varepsilon_n]\right)^\intercal\right].\end{equation} Since $\theta_n, \theta \in \mathbb{R}^d$, $\varepsilon_n$ in~\eqref{var} is a $d \times 1$ column vector and $\Sigma_{\varepsilon_n}$ is a $d \times d$ positive-definite covariance matrix. The ``marginal quantiles'' of the vector $\varepsilon_n : = (\varepsilon_{1,n}, \varepsilon_{2,n}, \ldots, \varepsilon_{d,n})$ are given by \begin{equation}\label{simulquant} \psi^{\text{q}}(P_{\varepsilon_n},\gamma) \equiv Q_{\varepsilon_n}(\gamma_1,\gamma_2,\ldots,\gamma_d) := \left(F_{1,n}^{-1}(\gamma_1),F_{2,n}^{-1}(\gamma_2),\ldots, F_{d,n}^{-1}(\gamma_d)\right),   \end{equation} where for each $i=1,2,\ldots,d$, $\gamma_s \in (0,1)$, $F_{i,n}$ is the cumulative distribution function (cdf) of the real-valued random variable $\varepsilon_{i,n}$, and $F_{i,n}^{-1}(\gamma_s):= \min\{x: F_{i,n}(x) \geq \gamma\}.$ The functional $\psi$ serves to assess the quality of the point estimator, which we call an \emph{assessment functional} to distinguish it from the functional $\theta$. \item[(B)] \emph{Constructing a $(1-\alpha)$ confidence region on $\theta \in \mathbb{R}^d$.} A (random) set $C_n \in \sigma(Y_1,Y_2,\ldots,Y_n)$ is a called a $(1-\alpha)$-confidence region on $\theta$ if it includes $\theta$ with probability $(1-\alpha)$ as $n \to \infty$, that is, \begin{equation}\label{conf} \lim_{n \to \infty}\mathbb{P}(\theta \in C_n) = 1-\alpha.\end{equation} 

\end{enumerate}

We present \emph{batching} as an easily implemented mechanism that consistently approximates $P_{\varepsilon_n}$, in the process accomplishing the activities in (A) and (B).

\subsection{Batching for Uncertainty Quantification}

The idea of batching data for inference is not new, and has been employed extensively in the contexts of confidence interval construction for the population mean and quantile, and for variance parameter estimation. For example, using batches within the classical time series context especially when constructing confidence intervals on the mean goes back to interpenetration samples by~\cite{1946mah}, the jacknife by~\cite{1949que}, pseudoreplication by~\cite{1969mcc}, and subsampling by~\cite{1969har}. The corresponding literature in the simulation context is also vast and dates back to~\cite{1963con},~\cite{1966mecmck}, and~\cite{1979fis}. These are precursors to the now mature methods to construct confidence intervals on the steady-state mean and quantiles using batched simulation output ---see~\cite{2007aleetal,meksch1984,1982sch,1982gly,2022passinyeh,2013calnak,2020huinak,2023suetal} for an entry into this stream of literature. An underlying premise of this paper is that batching is a much more general device than the existing literature suggests, and is better seen as a resampling device akin to bootstrapping~\citep{2012shatu} or subsampling~\citep{1999polromwol}. Accordingly, batching seamlessly facilitates (at least in principle) consistent approximation of $P_{\varepsilon_n}$ for accomplishing (A) and (B) in the previous section.

To see how, recall again the ``observed dataset'' $(Y_1,Y_2, \ldots,Y_n)$ in $(\mathcal{Y},\mathcal{A})$, and the empirical measure $$P_n(A) := n^{-1} \sum_{j=1}^n \delta_{Y_j}(A), \quad A \in \mathcal{A}$$ constructed from the observed dataset, where $\delta_{Y_j}(A)$ denotes the Dirac measure taking the value 1 if $Y_j\in A$ and 0 otherwise.  Fundamental to batching is a $\emph{batch}$ of contiguous observations from the dataset $(Y_1,Y_2,\ldots,Y_n)$, as shown in Figure~\ref{batching}. Group $(Y_1,Y_2,\ldots,Y_n)$ into $b_n$ possibly overlapping batches each of size $m_n$, the first of which consists of observations $Y_1,Y_2, \ldots, Y_{m_n}$, the second consisting of $Y_{d_n + 1}, Y_{d_n + 2}, \ldots, Y_{d_n + m_n}$, and so on, and the last batch consisting of $Y_{(b_n-1)d_n+1}, Y_{(b_n-1)d_n+2}, \ldots, Y_{n}$. (We have implicitly assumed that $(b_n-1)d_n + m_n = n$.) The quantity $d_n \geq 1$ represents the offset between batches, with the choice $d_n=1$ corresponding to ``fully-overlapping'' batches and any choice $d_n \geq m_n$ corresponding to ``non-overlapping'' batches. Notice then that the offset $d_n$ and the number of batches $b_n$ are related as $d_n = \frac{n-m_n}{b_n-1}.$ Now observe that the data in batches $1,2,\ldots,b_n$ can be used to construct the corresponding empirical measures: $$ P_{i,n}(A) := m_n^{-1} \sum_{j=1}^{m_n} \delta_{Y_{(i-1)d_n+j}}(A), \quad A \in \mathcal{A}; \,\, i=1,2,\ldots,b_n.$$ 

\begin{figure}[h]
\begin{tikzpicture}
\draw[thick,-{latex}] (0,0) -- (\textwidth,0);
\foreach \x in {0,0.3,...,13.75}
  \draw (\x,3pt) -- (\x,-3pt); 
\draw [thin, black,decorate,decoration={brace,amplitude=10pt,mirror},xshift=0.4pt,yshift=-2pt](0,0) -- (7.5,0) node[black,midway,yshift=-0.6cm] {\footnotesize batch $1$};
\draw [thin, black,decorate,decoration={brace,amplitude=10pt},xshift=0.4pt,yshift=2pt](4.5,0) -- (12,0) node[black,midway,yshift=0.6cm] {\footnotesize batch $2$};
\draw [thin, black,decorate,decoration={brace,amplitude=10pt,mirror},xshift=0.4pt,yshift=-2pt](9,0) -- (13.75,0) 
node[black,midway,yshift=-0.6cm] {\footnotesize batch $3$};
\draw[arrows=->,line width=.4pt](0,0)--(0.5,-0.9) node[black,yshift=-0.15cm] {\footnotesize $1$};
\draw[arrows=->,line width=.4pt](7.5,0)--(8,-0.9) node[black,yshift=-0.15cm] {\footnotesize $m_n$};
\draw[arrows=->,line width=.4pt](4.5,0)--(5,0.9) node[black,yshift=0.2cm] {\footnotesize $d_n +1$};
\draw[arrows=->,line width=.4pt](12,0)--(12.5,0.9) node[black,yshift=0.2cm] {\footnotesize $d_n + m_n$};
\draw[arrows=->,line width=.4pt](9,0)--(9.5,-0.9) node[black,yshift=-0.15cm] {\footnotesize $2d_n + 1$};
\end{tikzpicture}
\caption{The figure depicts partially overlapping batches. Batch 1 consists of observations $(Y_j, \ j = 1,2,\ldots,m_n)$; batch 2 consists of observations $(Y_j, \ j=d_n + 1, d_n + 2, \ldots, d_n + m_n)$, and so on, with batch $i$ consisting of $(Y_j, \ j=(i-1)d_n + 1, (i-1)d_n + 2, \ldots,(i-1)d_n+m_n).$ There are thus $b_n := \lfloor d_n^{-1}(n-m_n) \rfloor + 1$ batches in total, where $n$ is the size of the dataset.} \label{batching}
\end{figure}
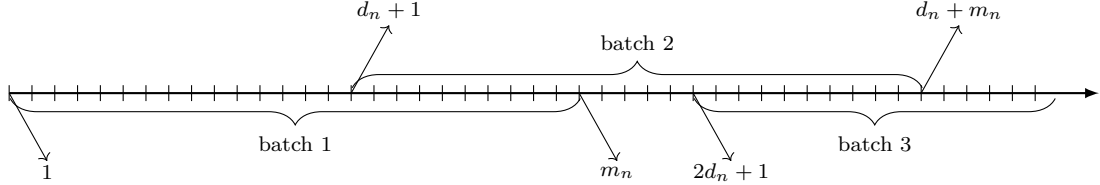 The crux of the batching idea for uncertainty quantification is that while the error $\varepsilon_n := \theta_n - \theta$ cannot be computed (because $\theta$ is unknown), under certain conditions it is ``well-approximated'' by the computable batch-error estimates \begin{equation}\label{errapp1}\varepsilon_{i,n} = \theta(P_{i,n}) - \theta_n, \quad i = 1,2,\ldots,b_n\end{equation} and \begin{equation}\label{errapp2} \tilde{\varepsilon}_{i,n} = \theta(P_{i,n}) - \bar{\theta}_n, \quad i = 1,2,\ldots,b_n,\end{equation} where $\bar{\theta}_n:=\frac{1}{b_n} \sum_{i=1}^{b_n} \theta(P_{i,n})$. The error approximations in~\eqref{errapp1} and~\eqref{errapp2} suggest two overlapping batch (OB) estimators of $P_{\varepsilon_n}$: \begin{align}\label{pbatch} P_{\mbox{\tiny OB-I},\varepsilon_n} &:= \frac{1}{b_n}\sum_{i=1}^{b_n} \delta_{{\varepsilon}_{i,n}}; \nonumber \\ P_{\mbox{\tiny OB-II},\varepsilon_n} &:= \frac{1}{b_n}\sum_{i=1}^{b_n} \delta_{\tilde{\varepsilon}_{i,n}}. \end{align}

It is important that the reader see why the measures $P_{\mbox{\tiny OB-I},\varepsilon_n}$ and $P_{\mbox{\tiny OB-II},\varepsilon_n}$ enable uncertainty quantification on $\varepsilon_n$. First, $P_{\mbox{\tiny OB-I},\varepsilon_n}$ and $P_{\mbox{\tiny OB-II},\varepsilon_n}$ are both observable, and thus $\psi(P_{\mbox{\tiny OB-I},\varepsilon_n})$ and $\psi(P_{\mbox{\tiny OB-II},\varepsilon_n})$ form ``plug-in'' estimators for $\psi(P)$. Second, upon identifying the weak limits of $P_{\mbox{\tiny OB-I},\varepsilon_n}$ and $P_{\mbox{\tiny OB-II},\varepsilon_n}$, a valid $(1-\alpha)$-confidence region on $\theta(P)$ can be constructed. Three more observations about the batching setup are noteworthy. \begin{enumerate} \item[(a)] The ``true'' error $\varepsilon_n := \theta(P_n) - \theta(P)$ is unknown since the probability measure $P$ in $\theta(P)$ is unknown. Batching constructs estimates $\varepsilon_{i,n}, i = 1,2,\ldots,b_n$ of $\varepsilon_n$ by replacing $P_n$ with $P_{i,n}$, and $P$ with $P_n$, in a manner reminiscent of resampling methods like bootstrapping~\citep{2012shatu,1992hall}.  \item[(b)] The measures $P_{\mbox{\tiny OB-I},\varepsilon_n}$ and $P_{\mbox{\tiny OB-II},\varepsilon_n}$ are ``random measures'' since they are each a function of the dataset $(Y_1,Y_2,\ldots,Y_n)$. \item[(c)] The random measures $P_{\mbox{\tiny OB-I},\varepsilon_n}$ and $P_{\mbox{\tiny OB-II},\varepsilon_n}$ are known but never explicitly computed (or stored) in practice. Instead, as we shall see, any desired statistical functional of $P_{\mbox{\tiny OB-I},\varepsilon_n}$ or $P_{\mbox{\tiny OB-II},\varepsilon_n}$ is directly computed from the available data.\end{enumerate}

\begin{remark}[Notational Abuse]\label{rem:scale}
    We use $P_{\mbox{\tiny OB-I},\varepsilon_n}$ and $P_{\mbox{\tiny OB-II},\varepsilon_n}$ (see~\eqref{pbatch}) for the proposed approximations to the desired probability measure $P_{\varepsilon_n}$. However, $\varepsilon_{i,n}$ and $\tilde{\varepsilon}_{i,n}$ need to be appropriately rescaled by $\sqrt{\frac{m_n}{n}}$ when comparing against the probability measure $P_{\varepsilon_n}$, implying that the notation.  $P_{\mbox{\tiny OB-I},\sqrt{\frac{m_n}{n}}\varepsilon_n}$ and $P_{\mbox{\tiny OB-II},\sqrt{\frac{m_n}{n}}\varepsilon_n}$ might be more appropriate. However, we choose against such notation on account of the resulting unwieldiness, but instead make the rescaling explicit when presenting the theoretical results in Sections~\ref{sec:estimatingpsi} and~\ref{sec:confreg}.
\end{remark}

\begin{table}[htb]
\caption{The table displays key notation used in the proposed batching procedures for performing uncertainty quantification using simulation output.}
\begin{tabularx}{\textwidth}{p{0.27\textwidth}X}
\toprule
  Object & Description \\
    \midrule
  $\theta \in \mathbb{R}^d$ & \emph{unknown} parameter, e.g., expected customer delay \\ 
  $Y_i \in \mathcal{Y}, \ i = 1,2,\ldots,n$ & \emph{known} output data obtained from simulation \\
  $P_n$ & \emph{known} empirical measure given the output data \\
  $\theta_n =\theta(P_n) \in \mathbb{R}^d$ & \emph{known} estimator of $\theta$ constructed from $Y_i, \ i=1,2,\ldots,n$, e.g., sample mean of customer delays\\
  $\varepsilon_n = \theta_n - \theta$ & \emph{unknown} error in the estimator $\theta_n$ with respect to $\theta$ \\
  $P_{\varepsilon_n}$ & \emph{unknown} distribution of error $\varepsilon_n$ \\
  $P_{\mbox{\tiny OB-I},\varepsilon_n}$, $P_{\mbox{\tiny OB-II},\varepsilon_n}$ & \emph{known} distributions constructed from batch estimates to approximate $P_{\varepsilon_n}$\\
  $\psi_{n} := \psi(P_{\varepsilon_n})$ & \emph{unknown} statistical functional of interest, e.g., variance of $\varepsilon_n$, bias of $\theta_n$ \\
  $\psi_{\mbox{\tiny OB-I},n} := \psi(P_{\mbox{\tiny OB-I},\varepsilon_n})$ & \emph{known} OB-I estimator of $\psi_n$ \\ $\psi_{\mbox{\tiny OB-II},n} := \psi(P_{\mbox{\tiny OB-II},\varepsilon_n})$ & \emph{known} OB-II estimator of $\psi_n$ \\
    \bottomrule
 \end{tabularx}
\end{table}

\begin{remark}[OB-I versus OB-II] Notice that OB-I uses $b_n$ batch estimates each computed with a batch of size $m_n$ along with a ``grand estimate center'' that is computed with the entire dataset of size $n$. The computational complexity of OB-I is thus $b_n c_0(m_n) + c_0(n)$, where $c_0(t)$ is the computational cost associated with constructing a point estimate using a batch of size $t$. OB-II is a computationally efficient alternative to OB-I that uses the average of individual batch estimates as centering, in the process avoiding the potentially large $c_0(n)$ term in OB-I's expression for computational complexity. In a nutshell, OB-II attempts to trade-off the quality of inference for computational efficiency, although a rigorous proof of OB-I's dominance in terms of the quality of inference is still open. \end{remark} 

\subsection{Related Terminology}
Inference in this paper pertains to simulation output uncertainty \emph{conditional on the given simulation model}. This is in contrast to another modern popular topic called simulation \emph{input uncertainty}~\citep{2003hen,2001chi,1997chehol,2012bar,2016lam}, which quantifies the effect of errors in the input distributions that form the primitives to the simulation. In effect, what we consider here provides a sense of how decision-making might be affected due to performing too few simulation runs, whereas input uncertainty deals with the corresponding effects due to a lack of adequate real-world data used when estimating the distributional input to the simulation.

Both output uncertainty and input uncertainty in simulation are subsumed by the recently phrased ``umbrella'' topic \emph{uncertainty quantification}~\citep{2021abdetal,2009naj,2017soi} which should be understood loosely as the effort to quantify the effect of all sources of error, e.g., input parameters, structure, logic, and solution, within models that include, but not limited to, simulation. Some examples of models other than simulation are stochastic differential equations~\citep{1986hoeporsto}, neural networks~\citep{2018botcurnoc}, and regression~\citep{2004was}. A number of other terms such as \emph{error estimation} and \emph{reliability estimation} have also come into use recently~\citep{2011barth} and should be carefully distinguished from what we see as the narrow topic of simulation output inference considered in this paper.

\subsection{Summary of Insight and Contribution} The context of this paper is the question of \emph{uncertainty quantification} on simulation output, that is, how to assess the error in a point estimator $\theta_n$ (of an unknown quantity $\theta$) that is constructed from simulation output data $(Y_1,Y_2,\ldots,Y_n)$. Uncertainty quantification on the error $\theta_n - \theta$ is essentially that of developing methods to approximate the unknown sampling distribution of $\theta_n - \theta$, using only $(Y_1,Y_2,\ldots,Y_n)$.

We present \emph{batching} as an omnibus principle for uncertainty quantification using simulation output. The key idea underlying batching is straightforward and stated in two steps: (i) create $b_n$ strategic sub-groups of the available data $(Y_1,Y_2,\ldots,Y_n)$ and construct a point estimate $\theta_{i,n}, i=1,2,\ldots,b_n$ from each sub-group; (ii) use the point estimates $\theta_{i,n}$ (and possibly the ``grand estimate'' $\theta_n$) to appropriately approximate the sampling distribution of $\theta_n - \theta$. Specific ways of performing (i) and (ii) result in the two methods OB-I and OB-II that we present in this paper. 

The qualitative contribution of this paper is the recognizing of batching as an omnibus inference device, that is, as a general principle for approximating the sampling distribution of $\theta_n-\theta$, thereby placing it alongside resampling ideas such as bootstrapping~\citep{1998efrtib,das2011,2012shatu} and subsampling~\citep{1999polromwol}. The quantitative contribution of this paper is fourfold.   \begin{enumerate} \item (strong consistency) We demonstrate that OB-I and OB-II are \emph{strongly consistent} in that they construct random measures which tend to the sampling distribution of $\theta_n - \theta$ almost surely in the supremum norm. \item (higher-order accuracy) We quantify the asymptotic distance between the sampling distribution of $\theta_n - \theta$ and the random measures constructed by OB-I and OB-II. In particular, OB-I and OB-II methods enjoy what has been called \emph{higher-order accuracy} --- OB-I and OB-II's random measures are closer to the sampling distribution of $\theta_n-\theta$ than the canonical distance to their common limit. 
\item (confidence regions) We characterize OB-I and OB-II distributions (along with percentile tables) that can be used to construct asymptotically valid confidence regions on the unknown parameter $\theta$.
\item (numerical experience) We provide a partial report on our now extensive numerical experience with uncertainty quantification using OB-I and OB-II. Such experience has been crucial for insight on selecting batch sizes and the extent of their overlap during implementation. We make OB-I and OB-II inference software codes publicly available.

\end{enumerate} 

Portions of our treatment of OB-I and OB-II have appeared elsewhere. For example, broad commentary on batching as an omnibus inference idea appear in~\citep{2023pasWSC,2023pasSW}, and specific  confidence region constructions with OB-I and OB-II appear in~\citep{2023suetal,2022passinyeh}. The principal content of this paper --- the formal treatment of OB-I and OB-II random measures as approximations of the desired sampling distribution --- has not appeared elsewhere. Likewise, Theorem~\ref{thm:batchstrcons},~Theorem~\ref{thm:hoa}, and Lemma~\ref{lem:bigbatch} on strong consistency and higher-order accuracy are new and have not appeared elsewhere.


\section{MATHEMATICAL PRELIMINARIES}

\label{sec:notation} In what follows, we list notation, define key ideas, standing assumptions invoked in the paper, and important results.
 
\subsection{Notation} (i) If $x \in \mathbb{R}^d$, then $x_j, j = 1,2,\ldots,d$ refers to the $j$-th coordinate of $x$, implying that $x= (x_1,x_2,\ldots,x_d)$. (ii) For $x,y \in \mathbb{R}$, $a \vee b := \max(a,b)$ refers to the maximum of $a$ and $b$, and $a \wedge b := \min(a,b)$ refers to the minimum of $a$ and $b$. (iii) $\|x\|_p := \left( \sum_{j=1}^d |x_j|^p\right)^{\frac{1}{p}}$ refers to the $L_p$ norm of $x \in \mathbb{R}^d$. As usual, we use the special notation $\|x\|$ for the $p=2$ case to refer to the $L_2$ norm of $x \in \mathbb{R}^d$. (iv) For a random sequence $\{X_n, n \geq 1\}$, we write $X_n \as X$ to refer to almost sure convergence (also known as convergence with probability one), $X_n \inP X$ to refer to convergence in probability, and $X_n \inD X$ to refer to convergence in distribution (also known as weak convergence).  (v) For positive valued sequences $\{a_n, n \geq 1\}$ and $\{b_n, n \geq 1\}$, we say $a_n \sim b_n$ to mean $a_n/b_n \to 1$ as $n \to \infty$. (vi) The Dirac measure $\delta_{x}$ is a probability measure on $\mathbb{R}^d$ (with any sigma algebra of subsets of $\mathbb{R}^d$) that satisfies $\delta_x(A) = 1$ if $x \in A$ and $0$ otherwise, for measurable sets $A$. (vii) For $t=(t_1,t_2,\ldots,t_d) \in \mathbb{R}^d,$ $(-\infty,t]$ refers to the ``southwest rectangle'' $(-\infty,t_1] \times (-\infty,t_2] \times \cdots (-\infty,t_d]$.

\subsection{Definitions}

\begin{definition}[Space $\mathcal{P}(\mathbb{R}^d)$ of Probability Measures]\label{defn:pbspace} The \emph{space of probability measures} $\mathcal{P}(\mathbb{R}^d)$ refers to the set of non-negative Borel measures $\mu$ on $\mathbb{R}^d$ such that $\mu(\mathbb{R}^d) = 1.$  

\end{definition}

\begin{definition}[Continuous Statistical Functional]\label{defn:contstatfn} A \emph{statistical functional} $\psi: \mathcal{P}(\mathbb{R}^d) \to \mathbb{R}$ is a real-valued map having domain $\mathcal{P}(\mathbb{R}^d)$, the space of probability measures on $\mathbb{R}^d$. The statistical functional $\psi$ is said to be \emph{continuous} at $\lambda_0 \in \mathcal{P}(\mathbb{R}^d)$ if $\|\lambda_n - \lambda_0\|_{\infty} \to 0 \,\, \Longrightarrow \,\, |\psi(\lambda_n) - \psi(\lambda_0)| \to 0,$ for any sequence $\{\lambda_n, n \geq 1\} \subset \mathcal{P}(\mathbb{R}^d)$. $\| \cdot \|_{\infty}$ refers to the uniform norm on $\mathcal{P}(\mathbb{R}^d)$: $$\|\lambda_1 - \lambda_2\|_{\infty} = \sup_{t \in \mathbb{R}^d} \bigg | \lambda_1(R_t) - \lambda_2(R_t)  \bigg |, \quad \lambda_1,\lambda_2 \in \mathcal{P}(\mathbb{R}^d),$$ where the rectangle $R_t := (-\infty,t_1] \times (-\infty,t_2] \times \cdots \times (-\infty, t_d],$ and $t = (t_1,t_2,\ldots,t_d).$   
\end{definition}

\begin{definition}[$d$-dimensional Wiener process, also called Brownian motion]\label{defn:brownian} Recall an $\mathbb{R}^d$-valued stochastic process $\mathbb{Y} = \{{Y}(t), t \in [0,\infty)\},$ defined on a space $\mathcal{P}(\mathbb{R}^d)$ of probability measures is called a \emph{Gaussian process}  if for any $0 \leq t_1 < t_2 < \cdots < t_n < +\infty$, $({Y}_i(t_1), {Y}_i(t_2), \ldots, {Y}_i(t_n))$ has a multivariate Gaussian distribution~\cite[pp. 4]{1980ser} for $i=1,2,\cdots,d$. The $d$-dimensional \emph{Wiener process} $\mathbb{W} = \{{W}(t), t \in [0,\infty)\} \subset \mathbb{R}^d$ is a special type of Gaussian process satisfying the following three properties:\begin{enumerate} \item[(a)] with probability one, the map $t \mapsto {W}(t)$ is continuous in $t$; \item[(b)] random vectors ${W}(t_2)-{W}(t_1), {W}(t_3) - {W}(t_2), \ldots, {W}(t_n)- {W}(t_{n-1})$ are independent for $0\leq t_1 < t_2 < \cdots < t_n < +\infty$; and \item[(c)] for $t,h \geq 0$ and fixed ${\mu} \in \mathbb{R}^d, \Sigma \in \mathbb{R}^{d \times d}$, ${W}(t+h) - {W}(t) \overset{\rm{d}}{=} N\left(h\mu,h\Sigma \Sigma^\intercal\right).$\end{enumerate} A $d$-dimensional Wiener process is called a \emph{standard Wiener process} if ${W}(0)=0, \mu=0$ and $\Sigma = I_d$. See Theorem 8.2 of \cite{1999bil} for more. 
\end{definition}

\subsection{Assumptions}

We state a number of key assumptions, some subset of which will be invoked by each result in the paper. The first three assumptions on stationarity, strong mixing, and CLT are standard assumptions. Further discussion on the stringency of the invoked assumptions appears in the commentary that follows each result.

\begin{assumption}[Stationarity]\label{ass:stationarity} The $\mathbb{R}^d$-valued sequence $\{Y_n, n \geq 1\}$ is stationary, that is, for any $n_j, j=1,2,\ldots,k < \infty$ and $k < \infty$, the distribution of $(Y_{n_1+\tau}, Y_{n_2+\tau}, Y_{n_3 + \tau}, \ldots,Y_{n_k+\tau})$ does not depend on $\tau \in \{1,2,\ldots \}$. 

\end{assumption}

\begin{assumption}[CLT]\label{ass:CLT}
The sequence $\{\theta_n, n \geq 1\}$ of estimators satisfies a central limit theorem, that is, \begin{equation}\label{clt} \sqrt{n}(\theta_n - \theta) \inD \sigma Z(0,\Sigma),\end{equation} where $Z(0,\Sigma)$ is normal random vector with mean zero and a covariance matrix $\Sigma$.
\end{assumption}

\begin{assumption}[Strong Mixing]\label{ass:phimixing} Suppose $\mathcal{G},\mathcal{H}$ are sub-$\sigma$-algebras of $\mathcal{F}$ in the probability space $(\Omega, \mathcal{F},P)$. Recall that the strong mixing constant $\alpha(\mathcal{G},\mathcal{H})$ is given by \begin{align} \alpha(\mathcal{G},\mathcal{H}) &= \sup_{A \in \mathcal{G}} \sup_{B \in \mathcal{H}} \left | P(AB) - P(A)P(B) \right| \nonumber \\ &= \frac{1}{2} \sup_{A \in \mathcal{G}} \mathbb{E}\left[ \left|P(A \vert \mathcal{H}) - P(A)\right|\right] \nonumber \\ &= \frac{1}{2} \sup_{A \in \mathcal{H}} \mathbb{E}\left[ \left|P(A \vert \mathcal{G}) - P(A)\right|\right].\end{align} We assume that the $S$-valued sequence $\{X_n, n \geq 1\}$ has strong mixing~\cite[pp. 347]{2009ethkur} constants $\alpha_n := \alpha(\mathcal{F}_{k},\mathcal{F}_{k,n})$ satisfying $\alpha_n \searrow 0$ as $n \to \infty$, where $\mathcal{F}_{k} : =\sigma(X_1,X_2,\ldots,X_k)$, $\mathcal{F}_{k,n} : =\sigma(X_{k+n},X_{k+n+1},\ldots)$ denote  sub-$\sigma$-algebras of $\mathcal{F}$ ``separated by $n$.''

\end{assumption}

\begin{assumption}[Strong Invariance]\label{ass:stronginvar} The sequence $\{\theta(P_n), n \geq 1\}$ of estimators satisfies the following strong invariance principle. On a rich enough 
probability space, there exists a constant $\sigma>0$, a standard Wiener process $\{W(t), t \geq 0\}$ and a stationary stochastic process $\{\tilde{Y}_n, n \geq 1\} \overset{d}{=} \{Y_n, n \geq 1\}$ such that as $n \to \infty$, \begin{equation}\label{inv} \sup_{0 \leq t \leq n} \left |\sigma^{-1}\left(\theta(P_{\lfloor t \rfloor}) - \theta(P)\right) - t^{-1}W(t)\right | \leq \Gamma \, n^{-1/2-\delta}\sqrt{\log n} \quad \emph{ a.s.},\end{equation} where the constant $\delta>0$ and the real-valued random variable $\Gamma$ satisfies $\mathbb{E}[\Gamma] < \infty$. \end{assumption}

Many of the results we present in this paper will not hold if the dependence within the underlying ``native time series'' $\{X_t, t \geq 1\}$ is too extreme. Assumption~\ref{ass:phimixing} is commonly used to exclude such contexts through an implicit comparison with a corresponding sequence consisting of iid random variables. Specifically, Assumption~\ref{ass:phimixing} states that events chosen from the sigma algebras $\mathcal{F}_k$ and $\mathcal{F}_{k,n}$ behave like independent events as the ``separation'' $n$ between the sigma algebras tends to infinity. 

Assumption~\ref{ass:stronginvar} on strong invariance is a statement about $\{\theta(P_n), n \geq 1\}$ ``looking like'' a Wiener process on a certain scaling. In effect, Assumption~\ref{ass:stronginvar} stipulates that the scaled process $\left\{\sqrt{n}\sigma^{-1}\left(\hat{\theta}_{\lfloor t \rfloor} - \theta(\tilde{P})\right), t \leq n\right\}$ can be approximated uniformly to within $n^{-\delta}$ \emph{almost surely}, by a suitable standard Wiener process on a rich enough probability space.  Assumption~\ref{ass:stronginvar} holds for a variety of weakly dependent processes as argued in~\cite{1975phisto}, and~\cite{1988glyigl}. Also see~\cite{1981csorev} for strong invariance theorems on partial sums, empirical processes, and quantile processes. 

The  value of $\delta >0$ in Assumption~\ref{ass:stronginvar} is largely dictated by the tail behavior of the marginal distribution associated with the underlying time series and the extent of the internal dependence. As noted in~\cite{2023suetal}, for instance, the Koml\'{o}s-Major-Tusn\'{a}dy approximation~\citep{1975kmt} implies that when $\{X_n, n \geq 1\}$ forms an independent and identically distributed (iid) real-valued sequence with $\mathbb{E}[|X_1|^p] < \infty$ for some $p > 2$, and $\hat{\theta}_n = n^{-1} \sum_{j=1}^n X_j$, Assumption~\ref{ass:stronginvar} holds with $\delta = \frac{1}{2} - \frac{1}{p}.$ For other examples where $\delta$ can be deduced, see~\cite{2014berliuwu} for the dependent context,~\cite{2015mer} for additive functionals of certain types of Markov chains, and~\cite{1996hesha} for M-estimators.

\begin{assumption}[Edgeworth Expansion, see~\cite{1992hall}]\label{ass:edgeworth} The sequence $\{\theta_n, n\geq 1\}$ satisfies an Edgeworth expansion \begin{align}\label{edge} P\bigg(\sqrt{n}\left(\theta_n - \theta\right) \leq x \bigg) &= \Phi(x) + n^{-1/2}p_1(x)\varphi(x) + n^{-1}p_2(x)\varphi(x) + \cdots \nonumber \\ & \hspace{1.8in} + n^{-j/2}p_j(x)\varphi(x) + \cdots, \end{align} where $p_j$ is a polynomial of degree no more than $3j-1$, odd for even $j$ and even for odd $j$. \end{assumption} It is thought that in a great many cases of practical interest, an Edgeworth expansion of the form~\eqref{edge} holds. See~\cite{1976rei} and~\cite[Chapter 2]{1992hall} for a treatment when $\theta_n$ is a sample quantile or a smooth function of a sample mean.

\subsection{Useful Results} We next state two useful results that will be invoked in the proofs of various results in the paper.

\begin{theorem}[Slutsky's Theorem,~\cite{1980ser}]\label{slutsky} Let $\{X_n, n \geq 1\}$ be a real-valued sequence such that $X_n\inD X$ where $X$ is a real-valued random variable. Also, let $\{Y_n, n \geq 1\}$ be a real-valued sequence such that $Y_n\inP c$ where $c$ is a constant. Then $X_n+Y_n\inD X+c$, $X_nY_n \inD cX$ and $X_n/Y_n \inD X/c$ if $c \neq 0$. 
\end{theorem}

\begin{theorem}[P\'{o}lya's Theorem, see~\cite{1980ser}]\label{polya}
Let $H_n, H, n=1,2,\ldots$ be cumulative distribution functions on $(\mathbb{R}^d,\mathfrak{B}(\mathbb{R}^d))$, and suppose that $H_n \inD H$ at $H$-continuity points, that is, for each $t \in \mathbb{R}^d$ where $H$ is continuous, $H_n(t) \to H(t)$ as $n \to \infty$. Then, the convergence is uniform, that is, $\sup_{t \in \mathbb{R}^d} |H_n(t) - H(t)| \to 0$ as $n \to \infty.$\end{theorem}

\section{EXAMPLE SIMULATION SETTINGS}\label{sec:examples}
To provide the reader a diversity of contexts that come under the purview of the methods described in this paper, and to provide further clarity on notation, we present three example settings. In each case, we identify the unknown parameter $\theta$, the estimator $\theta_n$, and the dataset $(Y_1,Y_2,\ldots,Y_n)$.

\subsection*{Example I (Wait Time in a $G/G/1$ Queue)} Consider the $G/G/1$ queue where a single server serves customers arriving according to an arrival process with independent and identically distributed (iid) inter-arrival times having distribution $G_1$. Customers are served in the order in which they arrive after joining a queue having infinite capacity. Service times for customers are iid according to a distribution $G_2$. Suppose $G_1$, $G_2$, and the initial conditions are such that the system is at steady-state, that is, $Y_n \overset{\mbox{\scriptsize d}}{=} Y \,\, \forall n \geq 1,$ where $Y$, waiting time of an arbitrary customer, is a well-defined random variable having distribution $F_W$. Let $\theta = \min\{w: F_{Y}(y) \geq 0.90\}$ denote the $0.9$-quantile of $Y$, and suppose that $(Y_1,Y_2,\ldots,Y_n)$ are the observed waiting times of the first $n$ customers in the system, so that $$\theta_n := \min\left\{y: F_{n,Y}(y) \geq 0.90\right\}; \quad F_{n,Y}(y) := \frac{1}{n}\sum_{j=1}^n \mathbb{I}(Y \leq y).$$ As described in the introduction, a simulationist interested in uncertainty quantification on $\theta_n$ is essentially attempting to understand the sampling distribution of $\epsilon_n = \theta_n-\theta$. 

\subsection*{Example II (Time-Dependent Inventory Levels in a Supply Chain)}

\begin{figure}[h]
\centering
\includegraphics[trim={0cm 2cm 0cm 1cm},clip,width = 0.75\textwidth,angle=0]{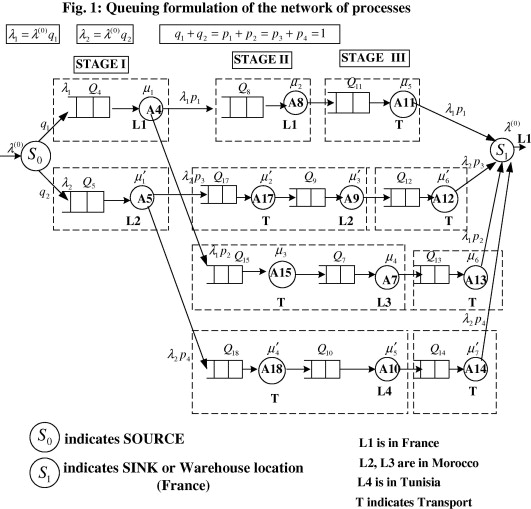}\caption{A queueing network model of a supply chain.}
\end{figure}
As a more elaborate example, consider the global supply chain introduced in the tutorial by~\cite{2014ing}, where the simulationist wishes to analyze the delivery of computing servers produced in Europe to the Asia-Pacific region, with the specific intention of evaluating whether it may be wise to move production to Singapore. Due to the complexity and scale of such a supply chain, it is easy to see why a simulation model would be helpful in answering many narrow questions, e.g., effect on inventory, effect on on-time delivery, effect on costs and revenue,  which together will be pertinent to the broader question of whether a move to Singapore is warranted.  

Consider one such narrow question, that of \emph{time-dependent inventory level}, that is, inventory as a function of time, at a specified location and observed over a horizon $[0,T$] of interest. The simulationist executes $n$ runs of the simulation, producing time-dependent inventory level $Y_{i}(t), t \in [0,T]$ during the $i$-th run. Importantly, notice that the $i$-th ``observation'' denoted $Y_i:= Y_{i}(t), t \in[0,T]$ is an entire function, in contrast with the previous example. As an aside, we note that we avoid the important and interesting question of how the infinite-dimensional object $Y_i$ should be stored in a digital computer, since it distracts from our main treatment. Suppose now that a simulationist who is especially interested in investigating low inventory levels chooses the parameter $\theta$ to be the $20$-th percentile inventory level as a function of time, that is, $\theta := \theta(t), t \in [0,T]$, where $\theta(t)$ is the $20$-th percentile inventory at time $t$. Thus, like the data $Y_i$, the parameter $\theta$ in this example is also function-valued. Recalling the ``dataset'' $(Y_1,Y_2,\ldots,Y_n)$ generated by $n$ runs of the simulation, an estimator $\theta_n : = {\theta}_n(t), t \in [0,T]$ of $\theta$ can then be constructed as:\begin{equation} {\theta}_n(t) := \min\left\{y: \frac{1}{n} \sum_{j=1}^n \mathbb{I}(Y_j(t)\leq y) \geq 0.2)\right\}, \quad t \in [0,T].\end{equation}  

The simulationist may have chosen a different parameter of interest, e.g., the mean vector of inventory levels at $d$ specific locations in the supply chain, at the fixed time instant $T$. In this case, denoting $\pi_{j,T}, j=1,2,\ldots,d$ as the inventory level distribution at time $T$ in location $j$, and denoting $Y_{i,j}(T), j=1,2,\ldots,d; i=1,2,\ldots, n$ as the $i$-th observed inventory level in location $j$ at time $T$, we can write: $$\theta := \left(\int_{\mathbb{R}} z \, \pi_{1,T}(\mbox{d}z), \int_{\mathbb{R}} z \, \pi_{2,T}(\mbox{d}z), \ldots, \int_{\mathbb{R}} z \, \pi_{d,T}(\mbox{d}z) \right).$$ Then, the estimator is ${\theta}_n := \left( \frac{1}{n}\sum_{i=1}^n Y_{i,1}(T),\frac{1}{n}\sum_{i=1}^n Y_{i,2}(T), \ldots, \frac{1}{n}\sum_{i=1}^n Y_{i,d}(T)\right).$

\subsection*{Example III (Nonlinear System of Equations)}  

Variable toll pricing has become a popular method to manage traffic on highways, by shifting purely discretionary traffic to off-peak hours or other roadways. Accordingly, a question of immense interest involves identifying the relationship between the toll price and the resulting congestion levels at steady state, toward better congestion pricing policies.

Suppose $p=(p_1,p_2,\ldots,p_d),\ p_i \in [0,M]$ represents the prevailing toll price for $d$ vehicle classes, and $\theta = \{(\theta_1(p), \theta_2(p), \ldots, \theta_d(p)),\ p \in [0,M]^d\}$ the corresponding expected steady state waiting time at the tolls for each of the $d$ classes given the toll prices specified by $p$. Given the complicated relationship between the expected wait time and the toll price, a simulation (whose mechanics are not relevant for our purposes) is used to estimate the parameter $\theta$. Suppose the simulation yields the output $(Y_1,Y_2, \ldots,Y_n)$, where $Y_i = (Y_{i1}(p), Y_{i2}(p), \ldots, Y_{id}(p)),\ p \in [0,M]^d$ represents the $i$-th realization of the wait time vector, that is, the vector wait times corresponding to the $i$-th vechicle in each of the $d$ classes, with $p$ held fixed. It is important to observe that each output observation $Y_i$ in this example is a \emph{random function or surface} of the toll price. A useful thought experiment that clarifies the nature of $Y_i$ is as follows. Fix and hold all ``random elements'' of the simulation while varying the toll price $p$ to form a time series of observations, each of which is a function of the price $p$.

Suppose the simulationist is interested in setting the tolls $p=(p_1,p_2, \ldots,p_d)$ so that the expected wait times for the $d$ classes matches target wait times $\gamma_1, \gamma_2, \ldots, \gamma_d$, respectively. Then the parameter $\theta$ is the solution (in $p$) to the following nonlinear system of equations: \begin{align}\label{nonlinsystem} \int_y y_j \pi_p (\mbox{d}y) &= \gamma_j, \quad j=1,2,\ldots,d. 
\end{align} Of course the solution $\theta$ to~\eqref{nonlinsystem} is unknown, but can be estimated as ${\theta}_n$ by solving the corresponding system  constructed using the data generated by simulation, that is, by solving the system: \begin{align}\label{nonlinsystemest} \frac{1}{n}\sum_{i=1}^n Y_{ij}(p) &= \gamma_j, \quad j=1,2,\ldots,d. 
\end{align} (There are existence and uniqueness issues pertaining to the solution of~\eqref{nonlinsystemest} but we omit discussion about such details here.) As in Example I and Example II, the inference question here is whether anything can be inferred about the nature of the error ${\theta}_n - \theta$.

\section{OB-I AND OB-II ESTIMATORS OF $\psi_n := \psi(P_{\varepsilon_n})$}\label{sec:estimatingpsi}
In this section, we present and analyze plug-in estimators $\hat{\psi}_{\mbox{\tiny OB-I},n}$, $\hat{\psi}_{\mbox{\tiny OB-II},n}$ of $\psi(P)$ constructed using the random measures $P_{\mbox{\tiny OB-I},\varepsilon_n}$ and $P_{\mbox{\tiny OB-II},\varepsilon_n}$. Recall that $\psi = (\psi_1,\psi_2,\ldots,\psi_p)$ is a vector of statistical functionals of interest, e.g., the expectation, variance, or simultaneous quantiles. Formally, \begin{align}\label{batchest}\hat{\psi}_{\mbox{\tiny OB-I},n} &:= \psi(P_{\mbox{\tiny OB-I},\varepsilon_n}) = \psi\left(\frac{1}{b_n}\sum_{i=1}^{b_n} \delta_{\varepsilon_{i,n}}\right);  \nonumber \\ \hat{\psi}_{\mbox{\tiny OB-II},n} &:= \psi(P_{\mbox{\tiny OB-II},\varepsilon_n}) = \psi\left(\frac{1}{b_n}\sum_{i=1}^{b_n} \delta_{\tilde{\varepsilon}_{i,n}}\right). \end{align} Table~\ref{table:psi} lists exact expressions for $\hat{\psi}_{\mbox{\tiny OB-I},n}$ and $\hat{\psi}_{\mbox{\tiny OB-II},n}$ (after appropriate rescaling) for some common choices of $\psi$. 

\begin{table}[htb]
\caption{The table displays OB-I and OB-II estimators for some commonly used assessment functionals $\psi$. Throughout the table, we use the notation $\varepsilon_{i,n} := \theta(P_{i,n}) - \theta_n, \tilde{\varepsilon}_{i,n} := \theta(P_{i,n}) - \bar{\theta}_n$ and $\bar{\varepsilon}_n := b_n^{-1}\sum_{i=1}^{b_n} \varepsilon_{i,n}$.  }
\begin{tabularx}{\textwidth}{p{0.1\textwidth}p{0.21\textwidth}p{0.32\textwidth}p{0.3\textwidth}}
\small

\toprule
 & Functional ($\psi$) & OB-I ($\hat{\psi}_{\mbox{\tiny OB-I},n}$) & OB-II ($\hat{\psi}_{\mbox{\tiny OB-II},n}$)\\ 
 \midrule
  bias & $\mathbb{E}[\varepsilon_n]$ & $\frac{1}{b_n}\sum_{i=1}^{b_n} \sqrt{\frac{m_n}{n}}\, \varepsilon_{i,n}$ & $\frac{1}{b_n}\sum_{i=1}^{b_n} \sqrt{\frac{m_n}{n}}\,\tilde{\varepsilon}_{i,n} = 0$ \\
  variance & $\mbox{Var}(\varepsilon_n)$ & $\frac{m_n}{n}\left(\frac{1}{b_n}\sum_{i=1}^{b_n} \varepsilon_{i,n}\varepsilon_{i,n}^\intercal - \bar{\varepsilon}_n\bar{\varepsilon}_n^\intercal\right)$ & $\frac{m_n}{n}\frac{1}{b_n}\sum_{i=1}^{b_n} \tilde{\varepsilon}_{i,n}\tilde{\varepsilon}_{i,n}^\intercal$ 
  \\
  quantiles & $\left(F_{j,n}^{-1}(\gamma_j),\ j=[d]\right)$ & $\left(\sqrt{\frac{m_n}{n}}\varepsilon_{\lfloor \gamma_jb_n\rfloor,j,n},\ j=[d]\right)$ & $\left(\sqrt{\frac{m_n}{n}}\tilde{\varepsilon}_{\lfloor \gamma_jb_n\rfloor,j,n},\ j=[d]\right)$
  \\
  \bottomrule
 \end{tabularx}\label{table:psi}
\end{table}

Our central question is whether $\hat{\psi}_{\mbox{\tiny OB-I},n}$ and $\hat{\psi}_{\mbox{\tiny OB-II},n}$ consistently estimate $\psi(P)$. We answer this question in a manner analogous to key results in the context of bootstrapping, e.g.,~\citep[Chapter 29]{das2011}. Specifically, Theorem~\ref{thm:batchstrcons} that follows guarantees the \emph{strong consistency} of $\hat{\psi}_{\mbox{\tiny OB-I},n}$ and $\hat{\psi}_{\mbox{\tiny OB-II},n}$, and the subsequent Theorem~\ref{thm:hoa} demonstrates that these estimators also enjoy a property that has been called higher-order accuracy in the context of bootstrapping~\citep{das2011}. 

\begin{theorem}[Strong Consistency of OB-I and OB-II]\label{thm:batchstrcons} Suppose Assumption~\ref{ass:stationarity} and~\ref{ass:CLT} hold, and Assumption~\ref{ass:phimixing} holds with strong-mixing constants $\alpha_n$ satisfying $\sum_{n=1}^{\infty} (n+1)\alpha_n^{\delta/(\delta+2)} < \infty$ for some $\delta>0$. Furthermore, the batch sizes $\{m_n, n \geq 1\}$ and number of batches $\{b_n, n \geq 1\}$ are such that as $n \to \infty$, \begin{equation} m_n \to \infty; \quad \frac{m_n}{n} \to 0; \tag{A.1}\end{equation} and \begin{equation} \sum_{n=1}^{\infty} b_n^{-2} < \infty. \tag{A.2}\end{equation} Then, for $\varepsilon_n := \theta_n - \theta(P)$, $\varepsilon_{i,n} := \theta(P_{i,n}) - \theta_n$, and $\tilde{\varepsilon}_{i,n} := \theta(P_{i,n}) - \bar{\theta}_n,$ as $n \to \infty$, \begin{equation}\label{strconob1} \sup_{t \in \mathbb{R}^d}\,\bigg | \frac{1}{b_n}\sum_{i=1}^{b_n} \mathbb{I}\bigg\{ \sqrt{m_n}\, \varepsilon_{i,n} \leq t \bigg\} - P\bigg(\sqrt{n} \, \varepsilon_n \leq t \bigg) \bigg | \as 0, \end{equation} and \begin{equation}\label{strconob2} \sup_{t \in \mathbb{R}^d}\,\bigg | \frac{1}{b_n}\sum_{i=1}^{b_n} \mathbb{I}\bigg\{ \sqrt{m_n}\, \tilde{\varepsilon}_{i,n} \leq t \bigg\} - P\bigg(\sqrt{n} \, \varepsilon_n \leq t \bigg) \bigg | \as 0. \end{equation} \qed
\end{theorem} We make a number of observations pertaining to Theorem~\ref{thm:batchstrcons} before providing a formal proof. \begin{enumerate} \item To see that Theorem~\ref{thm:batchstrcons} is indeed a strong consistency result on $P_{\mbox{\tiny OB-I},\varepsilon_n}$ and $P_{\mbox{\tiny OB-II},\varepsilon_n}$, notice that \begin{align*}P_{\mbox{\tiny OB-I},\varepsilon_n}\left((-\infty,\frac{t}{\sqrt{m_n}}]\right) &= \frac{1}{b_n}\sum_{i=1}^{b_n} \mathbb{I}\left\{ \sqrt{m_n}\, \varepsilon_{i,n} \leq t \right\}; \\ P_{\mbox{\tiny OB-II},\varepsilon_n}\left((-\infty,\frac{t}{\sqrt{m_n}}]\right) &= \frac{1}{b_n}\sum_{i=1}^{b_n} \mathbb{I}\left\{ \sqrt{m_n}\, \tilde{\varepsilon}_{i,n} \leq t \right\}. \end{align*} \item In~\eqref{strconob1} and~\eqref{strconob2},  $\varepsilon_{i,n}$ and $\tilde{\varepsilon}_{i,n}$ are each scaled-up by $\sqrt{m_n}$ whereas $\varepsilon_n$ is scaled up by $\sqrt{n}$. Such rescaling is necessary due to the ``small batch size'' condition in (A.1) which stipulates that $m_n$ should diverge slower than $n$. This implies that obtaining OB-I and OB-II estimates on $\psi(P_{\varepsilon_n})$ can be done by applying the same statistical functional to the distribution of $\sqrt{\frac{m_n}{n}}\,\varepsilon_{i,n},\ i=1,2,\ldots,b_n$ and $\sqrt{\frac{m_n}{n}}\,\tilde{\varepsilon}_{i,n},\ i=1,2,\ldots,b_n$, respectively. That we require $\frac{m_n}{n}\to 0$ for consistency also implicitly conveys that the rate at which the $\varepsilon_{i,n}$'s decay with $n$ is slower than the rate at which the estimator-errors $\varepsilon_n$ decay with $n$. \item The condition in (A.1) stipulates ``small batch sizes,'' that is, $\frac{m_n}{n} \to 0$ and $m_n \to \infty$, e.g., $m_n= n^{0.4}.$ The condition in (A.2) stipulates that $\sum_{n} b_n^{-2} < \infty$, e.g., $b_n = n^{0.6}.$ These conditions still leave many possible choices for $m_n$ and $b_n$. Are some of these better than others and if so, in what sense? It turns out that answering this question involves analyzing the mean squared error of $\hat{\psi}_{\mbox{\tiny OB-I},n}$ and $\hat{\psi}_{\mbox{\tiny OB-II},n}$ with respect to $\psi$, which in turn is a function of the extent of nonlinearity of both $\psi$ and $\theta$ as functionals. We do not go into further detail on this issue. \item The conditions (A.1) and (A.2) collectively encode the extent to which batches can be overlapped without relinquishing strong consistency. In particular, if the batches are overlapped too much, the condition $\sum_{n=1}^{\infty} (n+1)\alpha_n^{\delta/(\delta+2)} < \infty$ imposed on the strong-mixing constants $\alpha_n, n \geq 1$ may be violated. \end{enumerate}

A simple corollary of Theorem~\ref{thm:batchstrcons}, stated next without proof, is that if $\psi$ is a continuous functional, then the estimator $\hat{\psi}_n$ inherits strong consistency. 
\begin{corollary}[] Suppose $\psi: \mathcal{P}(\mathbb{R}^d) \to \mathbb{R}$ is continuous at $P$. (See Definition~\ref{defn:contstatfn}.) Then, under the postulates of Theorem~\ref{thm:batchstrcons}, as $n \to \infty$, \begin{equation} \left|\hat{\psi}_{\mathrm{\tiny OB-I},n} - \psi(P) \right | \as 0 \mbox{ and } \left|\hat{\psi}_{\mathrm{\tiny OB-II},n} - \psi(P) \right | \as 0 .\end{equation} \qed
\end{corollary}

We emphasize that the ``small batch'' condition stipulated by the second part of (A.1), i.e., $\frac{m_n}{n} \to 0$, is important for \emph{strong} consistency. In evidence, the following result asserts that even \emph{basic} consistency is lost if ``large batches'' are used. 

\begin{lemma}[Inconsistency of Large Batches]\label{lem:bigbatch} Suppose Assumption~\ref{ass:stationarity} and~\ref{ass:stronginvar} hold, and the batch sizes $\{m_n, n \geq 1\}$ and number of batches $\{b_n, n \geq 1\}$ are such that as $n \to \infty$, \begin{equation} \frac{m_n}{n} \to \beta > 0; \tag{A.3}\end{equation} and \begin{equation} b_n  \to \infty. \tag{A.4}\end{equation} Then as $n \to \infty$,
\begin{equation}\label{strcon} \sup_{t \in \mathbb{R}^d}\,\bigg | \frac{1}{b_n}\sum_{i=1}^{b_n} \mathbb{I}\bigg\{ \sqrt{m_n}\, \varepsilon_{i,n} \leq t \bigg\} - P\bigg(\sqrt{\Sigma} \, \tilde{B}(\beta) \leq t \bigg) \bigg | \as 0, \end{equation} where $\tilde{B}(\beta) := \frac{1}{\sqrt{\beta}} \frac{1}{1-\beta} \left( \int_{1-\beta}^1 W(u)\, \mathrm{d}u - \int_{0}^{\beta} W(u)\, \mathrm{d}u  - \beta (1-\beta)W(1)\right),$ and $\varepsilon_{i,n} = \theta(P_{i,n}) - \theta_n$. \qed
\end{lemma} 


Lemma~\ref{lem:bigbatch} implies a negative statement on consistency. To see why, notice that Assumption~\ref{ass:CLT} asserts that as $n \to \infty$, $P(\sqrt{n} \, \varepsilon_n \leq t) \to P (\sqrt{\Sigma}\, W(1) \leq t)$. However, by Lemma~\ref{lem:bigbatch}, $P_{\mbox{\tiny OB-I},\varepsilon_n}((-\infty,\frac{t}{\sqrt{m_n}}]) \to P(\sqrt{\Sigma}\, \tilde{B}(\beta) \leq t) \neq P (\sqrt{\Sigma}\, W(1) \leq t)$ if big batches, i.e., $\frac{m_n}{n} \to \beta>0$, are used. We will, however, later argue that big batches are desirable (and even crucial) when constructing confidence regions on $\theta(P)$.

Can anything be said about the rate at which $P_{\mbox{\tiny OB-I},\varepsilon_n}$ and $P_{\mbox{\tiny OB-II},\varepsilon_n}$ converge to $P_{\varepsilon_n}$? We answer this question using a result analogous to what has been called ``higher-order accuracy''~\cite{} in the context of the bootstrap.

\begin{theorem}[Higher-order Accuracy]\label{thm:hoa} Suppose Assumption~\ref{ass:stationarity} and Assumption~\ref{ass:edgeworth} hold. Suppose also that the condition (A.1) (on the batch size sequence $\{m_n, n \geq 1\}$) appearing in Theorem~\ref{thm:batchstrcons} holds. If the number of batches $\{b_n, n \geq 1 \}$ satisfies \begin{equation} \frac{b_n}{m_n} \,\, \to \,\, \infty, \tag{A.3}\end{equation} then, for $\varepsilon_n := \theta_n - \theta(P)$, $\varepsilon_{i,n} := \theta(P_{i,n}) - \theta_n$, and $\tilde{\varepsilon}_{i,n} := \theta(P_{i,n}) - \bar{\theta}_n,$ as $n \to \infty$, \begin{equation}\label{KS1} m_n^{1/2- \delta}\, \sup_{t \in \mathbb{R}^d}\,\bigg | \frac{1}{b_n}\sum_{i=1}^{b_n} \mathbb{I}\bigg\{ \sqrt{m_n}\, \varepsilon_{i,n} \leq t \bigg\} - P\bigg(\sqrt{n} \, \varepsilon_n \leq t \bigg) \bigg | \inP 0, \end{equation} and \begin{equation}\label{KS2} m_n^{1/2- \delta}\, \sup_{t \in \mathbb{R}^d}\,\bigg | \frac{1}{b_n}\sum_{i=1}^{b_n} \mathbb{I}\bigg\{ \sqrt{m_n}\, \tilde{\varepsilon}_{i,n} \leq t \bigg\} - P\bigg(\sqrt{n} \, \varepsilon_n \leq t \bigg) \bigg | \inP 0. \end{equation}
\end{theorem} 

Theorem~\ref{thm:hoa} broadly compares, on the $\sqrt{n}$-scaling, the random measures $P_{\mbox{\tiny OB-I},\varepsilon_n}$ and $P_{\mbox{\tiny OB-II},\varepsilon_n}$ produced by batching, against the true measure $P_{\varepsilon_n}$ that governs the error $\varepsilon_n = \theta(P_n) - \theta(P)$. Importantly, the theorem asserts that the supremum deviation between the random measures $P_{\mbox{\tiny OB-I},\varepsilon_n}$, $P_{\mbox{\tiny OB-II},\varepsilon_n}$ and $P_{\varepsilon_n}$ converge to zero faster than $\mathcal{O}(\frac{1}{\sqrt{n}})$, justifying the description ``higher order accurate.'' The implication of higher-order accuracy is interesting --- since each of $P_{\mbox{\tiny OB-I},\varepsilon_n}$ and $P_{\mbox{\tiny OB-II},\varepsilon_n}$ can be shown to converge to $P(\sqrt{\Sigma} W(1) \leq t)$ as $\mathcal{O}(\frac{1}{\sqrt{n}})$, higher-order accuracy implies that $P_{\mbox{\tiny OB-I},\varepsilon_n}$ and $P_{\mbox{\tiny OB-II},\varepsilon_n}$ are closer to $P_{\varepsilon_n}$ than to their respective limits. 

\subsection{Proofs of Theorem~\ref{thm:batchstrcons} and Theorem~\ref{thm:hoa}}

\begin{proof}[Proof of Theorem~\ref{thm:batchstrcons}]
From the triangular inequality, \begin{align}\label{check0} \MoveEqLeft \sup_{t \in \mathbb{R}^d}\left | \bar{M}_n(t) - P\bigg(\sqrt{n}\left(\theta(P_n) - \theta(P)\right) \leq t \bigg) \right | \leq \sup_{t \in \mathbb{R}^d} \left |\bar{M}_n(t) - P\bigg( \sqrt{\Sigma} \, W(1) \leq t \bigg) \right| \nonumber \\ & \hspace{0.8in} + \sup_{t \in \mathbb{R}^d}\,\left| P\bigg( \sqrt{\Sigma} \, W(1) \leq t \bigg) - P\bigg(\sqrt{n}\left(\theta(P_n) - \theta(P)\right) \leq t \bigg) \right|, \end{align} where $\bar{M}_n(t) := \frac{1}{b_n}\,\sum_{i=1}^{b_n} \mathbb{I}\bigg\{ \sqrt{m_n} \, \left( \theta(P_{i,n}) - \theta(P_n) \right) \leq t \bigg\}.$ 

Due to Assumption~\ref{ass:stationarity}, $\sqrt{m_n}(\theta(P_{i,n}) - \theta(P_n)), i=1,2,\ldots,b_n$ are identically distributed, and we have \begin{align}\label{mnbarsplit} \mathbb{E}\left[ \bar{M}_n(t) \right] &= \frac{1}{b_n}\, \sum_{i=1}^{b_n} P\bigg (\sqrt{m_n} \, \left( \theta(P_{i,n}) - \theta(P_n) \right) \leq t \bigg) \nonumber \\ &= P\bigg (\sqrt{m_n} \,\left( \theta(P_{1,n}) - \theta(P_n) \right) \leq t \bigg) =: \mu_n(t).\end{align} Also, due to Assumption~\ref{ass:phimixing}, Theorem~1 in~\cite{1978yos} assures us that there exists a constant $c_4 \in (0,\infty)$ such that for any $\epsilon>0$,\begin{equation}\label{bcprep} P\bigg(\left | \bar{M}_n(t) - \mu_n(t) \right | > \epsilon \bigg) \leq \frac{c_4}{b_n^2}.\end{equation} The inequality in~\eqref{bcprep} along with the condition in (A.2) implies that for any $\epsilon>0$, \begin{equation}\label{bccondsatisfied} \sum_{n=1}^{\infty} P\bigg(\left | \bar{M}_n(t) - \mu_n(t) \right | > \epsilon \bigg) < \infty,\end{equation} implying from Borel Cantelli's first lemma~\cite[pp. 59]{1995bil} that for each $t \in \mathbb{R}^d$, as $n \to \infty$, \begin{equation}\label{bcconc} \left | \bar{M}_n(t) - \mu_n(t) \right | \as 0.\end{equation} Next, we know from Assumption~\ref{ass:CLT} that as $n \to \infty$, \begin{equation} \label{estwkconv} \sqrt{m_n}\bigg( \theta(P_{1,n}) - \theta(P) \bigg) \,\,\inD \,\, \sqrt{\Sigma} \, W(1).\end{equation} From Assumption~\ref{ass:CLT} and the condition (A.1), we see that \begin{equation}\label{ass:rescale} \sqrt{m_n}\bigg(\theta(P) - \theta(P_n)\bigg) = \sqrt{\frac{m_n}{n}} \sqrt{n} \bigg(\theta(P)- \theta(P_n)\bigg) \inP 0.\end{equation} From~\eqref{estwkconv} and~\eqref{ass:rescale}, after applying Slutsky's theorem~\citep{1980ser}[pp. 19] (Theorem~\ref{slutsky}), we have $$\sqrt{m_n}\bigg(\theta(P_{1,n}) - \theta(P_n)\bigg) \,\, \inD \,\, \sqrt{\Sigma}\,W(1),$$ and hence for each $t \in \mathbb{R}^d$, \begin{equation}\label{nunconv} \left | \mu_n(t) - P\bigg(\sqrt{\Sigma} \, W(1) \leq t \bigg) \right | \to 0.\end{equation} From~\eqref{bcconc} and~\eqref{nunconv}, we see that for each $t \in \mathbb{R}^d$, as $n \to \infty$, \begin{equation}\label{ptwiselt}  \left | \bar{M}_n(t) - P\left(\sqrt{\Sigma} \, W(1) \leq t \right) \right | \as 0. \end{equation} Now apply P\'{o}lya's theorem (see Theorem~\ref{polya}) to get \begin{equation}\label{polya1}  \sup_{t \in \mathbb{R}^d}\,\left | \bar{M}_n(t) - P\left(\sqrt{\Sigma} \, W(1) \leq t \right) \right | \as 0,\end{equation} and \begin{equation}\label{polya2}  \sup_{t \in \mathbb{R}^d}\,\left | P\bigg(\sqrt{n} \, \left( \theta(P_{n}) - \theta(P) \right)) \leq t\bigg) - P\left(\sqrt{\Sigma} \, W(1) \leq t \right) \right | \to 0.\end{equation} Conclude from~\eqref{polya1},~\eqref{polya2} and~\eqref{check0} that the assertion of the theorem holds.
\end{proof}

\begin{proof}[Proof of Theorem~\ref{thm:hoa}] We provide a proof only for~\eqref{KS1}. A proof for~\eqref{KS2} follows along similar lines. Write \begin{align}\label{checkagain0} \MoveEqLeft\frac{1}{b_n} \sum_{i=1}^{b_n} \mathbb{I}\left\{\sqrt{m_n} \, \bigg( \theta(P_{i,n}) - \theta(P_n) \bigg) \leq t\right\} - P\left(\sqrt{n} \left(\theta(P_n) - \theta(P) \right) \leq t \right) &  \nonumber \\ & \hspace{1in}= \bar{M}_n(t) - \mu_n(t) + \mu_n(t) - P(\sqrt{n} \left(\theta(P_n) - \theta(P) \right) \leq t), \end{align} where we have adopted notation from the proof of Theorem~\ref{thm:batchstrcons}: \begin{align} \bar{M}_n(t) &:= \frac{1}{b_n}\sum_{i=1}^{b_n} \mathbb{I}\bigg\{ \sqrt{m_n} \, \left( \theta(P_{i,n}) - \theta(P_n) \right) \leq t \bigg\}; \nonumber \\
\mu_n(t) &:= \mathbb{E}[\bar{M}_n(t)] = P\bigg( \sqrt{m_n} \, \left( \theta(P_{1,n}) - \theta(P_n) \right) \leq t \bigg).
\end{align}  We will handle each of the sequences $\{\bar{M}_n(t) - \mu_n(t), n \geq 1\}$ and $\left\{\mu_n(t) - P(\sqrt{n} \left(\theta(P_n) - \theta(P) \right) \leq t), n \geq 1\right\}$ appearing in~\eqref{checkagain0} separately. Let's handle $\left\{\mu_n(t) - P\left(\sqrt{n} \left(\theta(P_n) - \theta(P) \right) \leq t \right), n \geq 1\right\}$ first by writing the Edgeworth expansion~\citep{1992hall}[Chapter 2] for each of $P\left( \sqrt{n} \, \left( \theta(P_{n}) - \theta(P) \right) \leq t \right)$ and $P\left(\sqrt{m_n} \, \left( \theta(P_{i,n}) - \theta(P_n) \right) \leq t\right)$ as follows. As $n \to \infty$, \begin{equation}\label{edgesampledist} P\left( \sqrt{n} \, \bigg( \theta(P_{n}) - \theta(P) \bigg) \leq t \right) = \Phi(t) + n^{-1/2}p_1(t)\varphi(t) + o(n^{-1}),\end{equation} where \begin{equation}\label{pone} p_1(t) = -\left\{k_{1,2} + \frac{1}{6}k_{3,1}(t^2-1)\right\},\end{equation} and \begin{equation}\label{edgeestsampledist} P\left( \sqrt{m_n} \, \bigg( \theta(P_{i,n}) - \theta(P_n) \bigg) \leq t \right) = \Phi(t) + m_n^{-1/2}\tilde{p}_1(t)\varphi(t) + o(m_n^{-1}),\end{equation} where \begin{equation}\label{pone2} \tilde{p}_1(t) = -\left\{k_{1,2} + \frac{1}{6}k_{3,1}(t^2-1) + \mathcal{O}(\sqrt{\frac{m_n}{n}}) \right\}.\end{equation} (The asymptotic expansions in~\eqref{edgesampledist} and~\eqref{edgeestsampledist} hold uniformly in $t \in \mathbb{R}$.) We see from~\eqref{edgesampledist} and~\eqref{edgeestsampledist} that \begin{align} \label{probseq} \MoveEqLeft m_n^{1/2 - \delta} \sup_{t \in \mathbb{R}^d}\, \bigg\{\mu_n(t) - P\left(\sqrt{n} \left(\theta(P_n) - \theta(P) \right) \leq t \right) \bigg\} \nonumber \\ & \hspace{1in} = m_n^{1/2 - \delta} \, \left(m_n^{-1/2} - n^{-1/2}\right)\, \sup_{t \in \mathbb{R}^d}\left\{p_1(t)\varphi(t) + \mathcal{O}(n^{-1/2})\right\} \nonumber \\ & \hspace{1in} \to 0 \mbox{ as } n \to \infty. \end{align} 

Next, we handle $$\bar{M}_n(t) - \mu_n(t) := \frac{1}{b_n} \sum_{i=1}^{b_n} \left[\underbrace{\mathbb{I}\bigg\{ \sqrt{m_n}(\theta(P_{i,n}) - \theta(P_n)) \leq t \bigg\} - P\bigg(\sqrt{m_n}(\theta(P_{i,n}) - \theta(P_n)) \leq t \bigg)}_{Z_n(t)}\right].$$ Note that the sequence $Z_n(t), n =1,2,\ldots,b_n$ is mean-zero and 
$\sup_{t \in \mathbb{R}^d}|Z_n(t)| \leq 2$. Furthermore, we can show that the functions $Z_n(\cdot), n = 1,2,\ldots,b_n$ lie in a sufficiently smooth Banach space implying that from Theorem 2.1 in~\citep{2015dedmer} there exists a constant $c_{0} < \infty$ such that for any $\epsilon>0$, \begin{equation}\label{smban} P\left( \sup_{t \in \mathbb{R}^d} \left | \bar{M}_n(t) - \mu_n(t) \right | > \epsilon \right) \leq \frac{c_{0}}{\epsilon^2}\frac{1}{b_n^2}.\end{equation} Now notice that due to~\eqref{probseq} and~\eqref{smban}, \begin{align}\label{martingalehoa} \MoveEqLeft P\left( m_n^{1/2 - \delta} \sup_{t \in \mathbb{R}^d} \left | \bar{M}_n(t) - P\left(\sqrt{n} \, \varepsilon_n \leq t \right) \right | > \epsilon \right) \nonumber \\ & \leq  P\left(m_n^{1/2 - \delta} \sup_{t \in \mathbb{R}^d} \left | \bar{M}_n(t) - \mu_n(t) \right | > \frac{\epsilon}{2}\right) \nonumber \\ & \hspace{1in} + P\left(m_n^{1/2 - \delta} \sup_{t \in \mathbb{R}^d} \left | \mu_n(t) -  P\left(\sqrt{n} \left(\theta(P_n) - \theta(P) \right) \leq t \right) \right | > \frac{\epsilon}{2}\right) \nonumber \\ & \leq \frac{16c_0}{\epsilon^4}\frac{m_n^{2-4\delta}}{b_n^2} + P\left(m_n^{1/2 - \delta} \sup_{t \in \mathbb{R}^d} \left | \mu_n(t) -  P\left(\sqrt{n} \left(\theta(P_n) - \theta(P) \right) \leq t \right) \right | > \frac{\epsilon}{2}\right) \nonumber \\ & \to 0 \mbox{ as } n \to \infty,\end{align} where the last line follows due to the condition in (A.3) and~\eqref{probseq}. Conclude that the assertion of the theorem holds.
\end{proof}

\section{OB-I and OB-II CONFIDENCE REGIONS}\label{sec:confreg} We now present OB-I and OB-II $(1-\alpha)$-confidence regions on $\theta$ constructed in a manner reminiscent of classical Student's $t$ confidence regions, by characterizing the weak limit of $\sqrt{n}\,\varepsilon_n$ after appropriate studentization. Specifically, consider 
\begin{align} T_{\mbox{\tiny OB-I},n} &:= \sqrt{n}\, \sqrt{\Sigma^{-1}_{\mbox{\tiny OB-I},n}} \, \varepsilon_n, \mbox{ where }  \Sigma_{\mbox{\tiny OB-I},n} := \frac{m_n}{b_n} \sum_{i=1}^{b_n} \varepsilon_{i,n} \, \varepsilon_{i,n}^{\intercal}; \nonumber \\
T_{\mbox{\tiny OB-II},n} &:= \sqrt{n}\, \sqrt{\Sigma^{-1}_{\mbox{\tiny OB-II},n}} \, \varepsilon_n, \mbox{ where }  \Sigma_{\mbox{\tiny OB-II},n} := \frac{m_n}{b_n} \sum_{i=1}^{b_n} \tilde{\varepsilon}_{i,n} \, \tilde{\varepsilon}_{i,n}^{\intercal}.\end{align} Suppose $T_{\mbox{\tiny OB-I},n}$ and $T_{\mbox{\tiny OB-II},n}$ converge weakly to characterizeable distribution-free weak limits $T_{\mbox{\tiny OB-I}}$ and $T_{\mbox{\tiny OB-II}}$, that is,  $$T_{\mbox{\tiny OB-I},n} \,\,\inD \,\, T_{\mbox{\tiny OB-I}} \mbox{ and } T_{\mbox{\tiny OB-II},n} \,\,\inD \,\, T_{\mbox{\tiny OB-II}}.$$ Then the corresponding $(1-\alpha)$-confidence ellipsoids on $\theta$ are given by \begin{equation}\label{ob1cr} \mathcal{C}_{\mbox{\tiny OB-I},p,1-\alpha} := \left\{x \in \mathbb{R}^d: \left\| \sqrt{n} \, \sqrt{\Sigma^{-1}_{\mbox{\tiny OB-I},n}} \left(x - \theta_n \right) \right\|_p \leq   t_{\mbox{\tiny OB-I},p,1-\alpha} \right\}\end{equation} and \begin{equation}\label{ob2cr} \mathcal{C}_{\mbox{\tiny OB-II},p,1-\alpha} := \left\{x \in \mathbb{R}^d: \left\| \sqrt{n} \, \sqrt{\Sigma^{-1}_{\mbox{\tiny OB-II},n}} \left(x - \bar{\theta}_n \right) \right\|_p \leq   t_{\mbox{\tiny OB-II},p,1-\alpha} \right\},\end{equation} where $t_{\mbox{\tiny OB-I},p,1-\alpha}$ and $t_{\mbox{\tiny OB-II},p,1-\alpha}$ are the $(1-\alpha)$ quantiles of random variables $\big\|T_{\mbox{\tiny OB-I}}\big\|_{p}$ and $\big\|T_{\mbox{\tiny OB-II}}\big\|_{p}$, respectively. Then, as long as the random variables $T_{\mbox{\tiny OB-I}}$ and $T_{\mbox{\tiny OB-II}}$ can be characterized (in the sense of being able to calculate any desired percentiles), the regions $\mathcal{C}_{\mbox{\tiny OB-I},1-\alpha}$ and $\mathcal{C}_{\mbox{\tiny OB-II},1-\alpha}$ are well-defined, computable, and asymptotically valid: $$\lim_{n \to \infty} P(\theta \in \mathcal{C}_{\mbox{\tiny OB-I},1-\alpha}) = 1-\alpha \mbox{ and } \lim_{n \to \infty} P(\theta \in \mathcal{C}_{\mbox{\tiny OB-II},1-\alpha}) = 1-\alpha.$$ Theorem~\ref{thm:ob1largebatch} and Theorem~\ref{thm:ob2largebatch} that follow characterize the weak limits $T_{\mbox{\tiny OB-I}}$ and $T_{\mbox{\tiny OB-II}}$. We omit proofs of Theorem~\ref{thm:ob1largebatch} and Theorem~\ref{thm:ob2largebatch} since they follow as a straightforward extension of the one-dimensional versions detailed in~\cite{2023suetal}.

\begin{theorem}[OB-I,~\cite{2023suetal}]\label{thm:ob1largebatch} Suppose that  Assumption~\ref{ass:stronginvar} holds, and that $\beta = \lim_{n \to \infty} \frac{m_n}{n} \in (0,1).$ Assume also that $b_n \to b \in \{2,3,\ldots,\infty\}$ as $n \to \infty$. Define \begin{numcases}{\chi^2_{{\mathrm{\tiny OB-I}}}(\beta,b) :=} \frac{1}{\beta(1-\beta)} \int_{0}^{1-\beta} \left( W(u+\beta) - W(u) - \beta W(1)\right)^{*2} \, du & $b = \infty$; \nonumber \\ \frac{1}{\beta b}\sum_{j=1}^{b} \left( W(c_j+\beta) - W(c_j) - \beta W(1)\right)^{*2} & $b \in \mathbb{N}\setminus \{1\}$, \nonumber \\ \end{numcases} where $a^{*2}:= aa^{\intercal}$ for $a \in \mathbb{R}^d$, and $c_j := (j-1)\frac{1-\beta}{b-1}$. Then, as $n \to \infty$, \begin{equation}\label{swag} \sqrt{\Sigma_{{\mathrm{\tiny OB-I},n}}} \inD \sqrt{\Sigma} \, \chi_{{\mathrm{\tiny OB-I}}}(\beta,b); \mbox{ and } \ T_{{\mathrm{\tiny OB-I},n}} \inD  \left(\chi_{{\mathrm{\tiny OB-I}}}(\beta,b)\right)^{-1}W(1) =: T_{{\mathrm{\tiny OB-I}}}.\end{equation} \qed \end{theorem}

Theorem~\ref{thm:ob2largebatch} is analogous to Theorem~\ref{thm:ob1largebatch} for the OB-II context and characterizes the random variable $T_{\mbox{\tiny OB-II}}$. 

\begin{theorem}[OB-II,~\cite{2023suetal}] \label{thm:ob2largebatch} Suppose that Assumption~\ref{ass:stronginvar} holds, and that $\beta := \lim_{n \to \infty} \frac{m_n}{n} >0.$ Assume also that $b_n \to b \in \{2,3,\ldots,\infty\}$ as $n \to \infty$. Define \begin{numcases}  {\chi^2_{{\mathrm{\tiny OB-II}}} (\beta,b) :=} \frac{\beta^{-1}}{1- \beta} \int_0^{1-\beta} \left( \tilde{W}_u(\beta) - \frac{1}{1-\beta} \int_{0}^{1-\beta} \tilde{W}_s(\beta)  \, ds \right)^{*2} du & $b = \infty$; \nonumber \\ \frac{1}{\beta} \frac{1}{b} \sum_{j=1}^{b} \left(\tilde{W}_{c_j}(\beta) -  \frac{1}{b} \sum_{i=1}^{b} \tilde{W}_{c_i}(\beta)  \right)^{*2}  & $b \in \mathbb{N}\setminus {1},$ \nonumber \\\end{numcases} where $a^{*2}:= aa^{\intercal}$ for $a \in \mathbb{R}^d$, $\tilde{W}_x(\beta) := W(x+\beta) - W(x), x \in [0,1-\beta]$, $\{W(t), t \in [0,1]\}$ is the standard Brownian motion, and $c_i := (i-1)\frac{1-\beta}{b-1}, i = 1,2, \ldots, b$.
Then, as $n \to \infty$, \begin{equation}\label{swag-samplemean} \sqrt{\Sigma_{{\mathrm{\tiny OB-II}}}} \inD \sqrt{\Sigma} \, \chi_{{\mathrm{\tiny OB-II}}}(\beta,b);\end{equation} and  \begin{numcases} {T_{{\mathrm{\tiny OB-II}}}(m_n,b_n) \inD } \left(\chi_{{\mathrm{\tiny OB-II}}}(\beta,b)\right)^{-1}\,\frac{\beta^{-1}}{1-\beta} \int_0^{1-\beta} \left(W(s+\beta) - W(s)\right) \, ds  & $b = \infty;$ \nonumber \\ \left(\chi_{{\mathrm{\tiny OB-II}}}(\beta,b)\right)^{-1}\frac{1}{\beta} \frac{1}{b}\sum_{i=1}^b W(c_i + \beta) - W(c_i)  & $b \in \mathbb{N} \setminus 1,$ \nonumber \\ \end{numcases} where $c_i := (i-1)\frac{1-\beta}{b-1}, i = 1,2, \ldots, b$. \qed
\end{theorem} 

We end this section with a few observations pertaining to Theorem~\ref{thm:ob1largebatch} and \ref{thm:ob2largebatch}.
\begin{enumerate} \item The random variables $T_{\mbox{\tiny OB-I}}$ and $T_{\mbox{\tiny OB-II}}$ are \emph{distribution-free} in the sense that they have no unknown parameters. Furthermore, they are \emph{implementable} in that fast code for percentile tables associated with functionals of $T_{\mbox{\tiny OB-I}}$ and $T_{\mbox{\tiny OB-II}}$ are now available through \texttt{https://web.ics.purdue.edu/\(\sim\)pasupath/solvers.html}. \item Theorem~\ref{thm:ob2largebatch} and Theorem~\ref{thm:ob1largebatch} pertain only to the large batch context, that is, the context where $\frac{m_n}{n} \to \beta>0$. The corresponding weak limits for the small batch context $\beta=0$ turns out to be Gaussian. \item When estimating $\psi(P)$, Lemma~\ref{lem:bigbatch} argued that big batches should not be used because even basic consistency is relinquished. Interestingly, our extensive numerical experience reveals that the opposite is true when constructing confidence regions. Specifically, large batches are preferred since they ensure rapid convergence to the nominal coverage probability. 

\item Another prominent batching variant arises due to using the \emph{weighted area estimator}~\citep{1983sch,2007aleetal,1990golsch,1990golmeksch} in place of $\Sigma_{{\mathrm{\tiny OB-I}}}$ and $\Sigma_{{\mathrm{\tiny OB-II}}}$. See~\citep{2023suetal} for the corresponding weak limit and also for variants that result from using small batch sizes, that is, $m_n$ such that $\frac{m_n}{n} \to 0$. 
\end{enumerate}

\section{NUMERICAL EXPERIMENTS}
In this section we study the performance of inference with batching in two numerical experiments. In the first experiment, we use a simple simulation that generates i.i.d. output data (observations) while in the second experiment we test our inference with batching using time-dependent output data of a single long simulation run. 

We claim that batching can be treated as another powerful device for inference, akin to resampling techniques. However, we do not intend to make comparisons in this paper and leave that for future research.
 
\subsection{Batching Inference with I.I.D. Observations: Tail of a Gamma Distribution}

In our first experiment we test the ability of OB-I and OB-II to approximate the sampling distribution of $\varepsilon_n= \theta(P_n) - \theta(P)$ in the service of uncertainty quantification. We estimate the 0.99-quantile of a gamma distribution from i.i.d. observations $Y_1,Y_2,\cdots,Y_n\sim \text{Gamma}(1,100)$ whose p.d.f. we denote by $f_Y(\cdot)$, implying (by CLT) that as $n \to \infty$, \begin{equation}\label{lt}\sqrt{n}\epsilon_n\inD\mathcal{N}\left(0,\frac{0.99\times0.01}{f_Y^2(\theta)}\right),\end{equation} where $\theta=\min\{y:\Pr(Y>y)=0.99\}$. The limiting distribution in~\eqref{lt} is depicted by a thin gray curve in Figure~\ref{fig:gamma-ex}, while the sampling distribution of $\varepsilon_n$ appears as the thick black curve. 

The methods proposed in this paper approximate the sampling distribution of $\sqrt{n}\varepsilon_n$ with batching instead of with the limiting distribution. The essential point being that when $n$ is small, the sampling distribution of $\sqrt{n}\epsilon_n$ may exhibit skewness not captured by the limiting distribution. Figure~\ref{fig:gamma-ex} illustrates this point and batching's attempt to approximate the distribution of $\varepsilon_n$. 
\begin{figure}[h]
\centering
    \includegraphics[width=\textwidth]{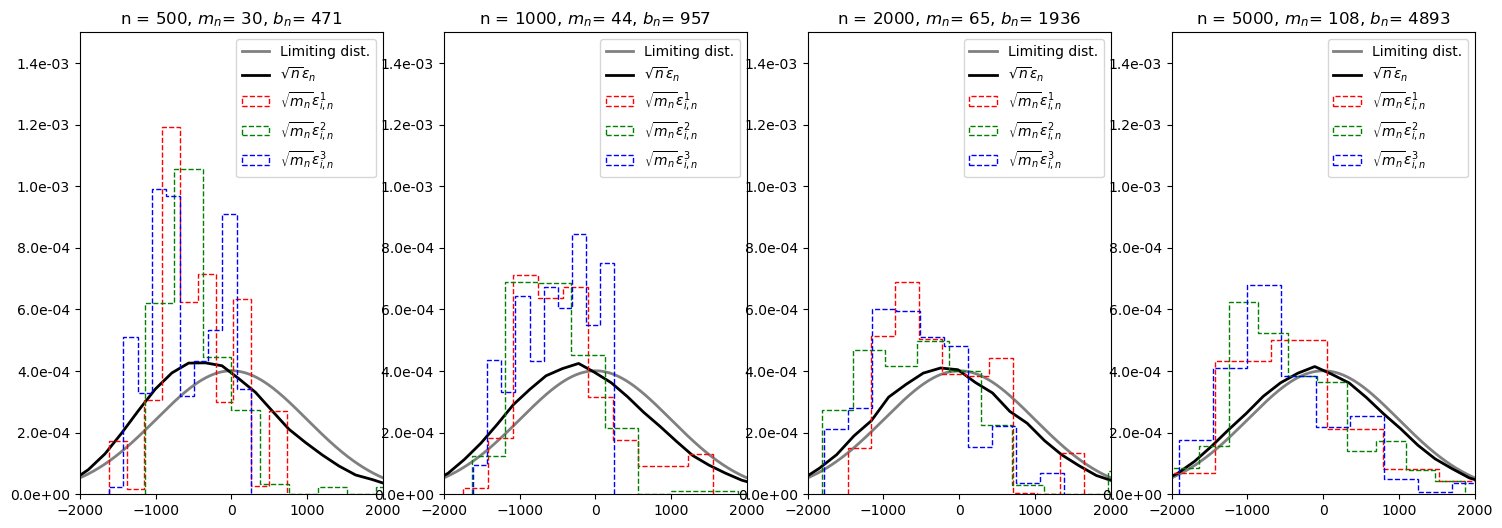}
    \caption{Histogram of three macro-replications of sampling distribution estimator with  $m_n=\sqrt{n}$ and $d_n=1$ (i.e., fully-overlapping batches) when the target $\theta(P)$ is the 99\% quantile of a gamma distribution.}
    \label{fig:gamma-ex}
\end{figure}

If the batch size $m_n$ in the previous experiment is $\mathcal{O}(n)$, for example, $m_n=0.2n$, then Lemma~\ref{lem:bigbatch} suggests that the consistency of the OB estimators is lost. This is illustrated in Figure~\ref{fig:gamma-ex2} where approximation suffers even with large $n$. 
On the other hand, if the goal is to construct confidence intervals that are asymptotically valid, large batches are preferred. This is illustrated through Table~\ref{tab:gamma-perf}.
\begin{figure}[h]
\centering
    \includegraphics[width=\textwidth]{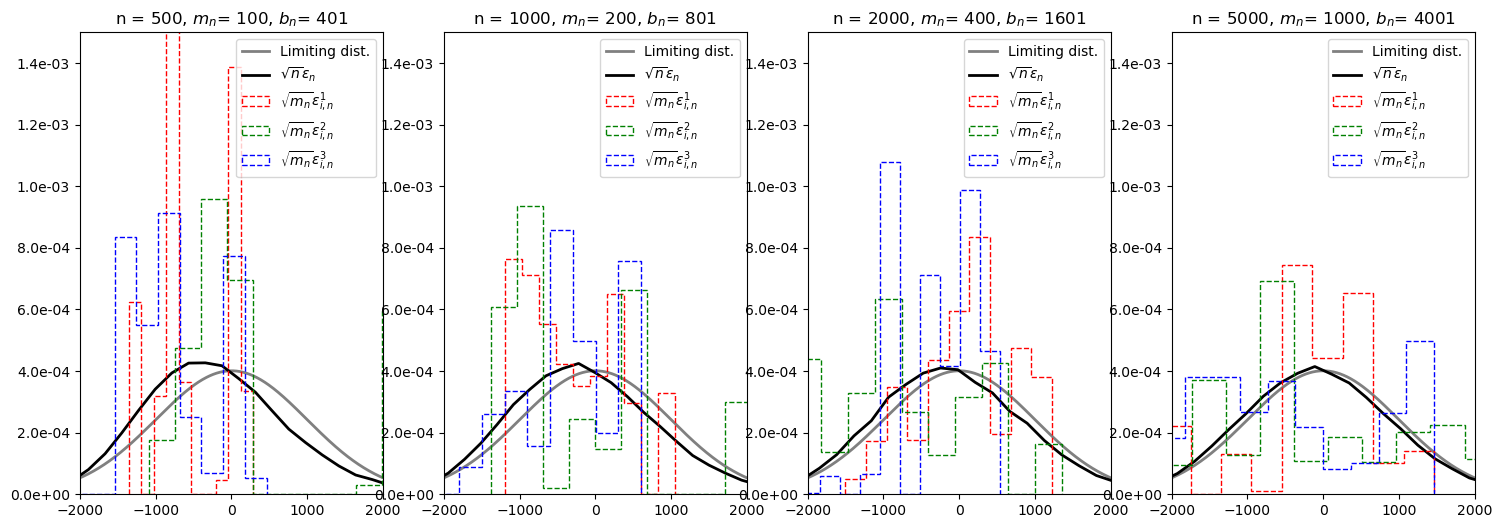}
    \caption{Histogram of three macro-replications of sampling distribution estimator with batching with $m_n=0.2n$ and $d_n=1$ (i.e., fully-overlapping batches) when the target $\theta(P)$ is the 99\% quantile of a gamma distribution. With growing number of observations $n$, we observe that large batch sizes struggle to consistently track the sampling distribution.}
    \label{fig:gamma-ex2}
\end{figure}

\begin{table}[h] 
\caption{Estimation quality of 0.99-quantile of a gamma distribution with 1000 macro-replications}
\label{tab:gamma-perf}
\small
\footnotesize
\scriptsize 
\tiny 
\centering
\begin{tabularx}{\linewidth}{X *{13}{S[table-format=2.1]}}
\toprule
& \multicolumn{4}{c}{CI with $m_{n} = 0.2n$} &  \multicolumn{8}{c}{RMSE of $\hat{\psi}_n$ with $m_{n} = \sqrt{n}$} \\
\cmidrule(lr){2-5}
\cmidrule(lr){6-13}
 & \multicolumn{4}{c}{coverage and half-width} & \multicolumn{2}{c}{bias} & \multicolumn{2}{c}{standard deviation} & \multicolumn{4}{c}{0.8-quantile} \\
 \cmidrule(lr){2-5}
 \cmidrule(lr){6-7}
 \cmidrule(lr){8-9}
 \cmidrule(lr){10-13}
 $n$  & {FOB-I} & {NOB-I} & {FOB-II} & {NOB-II} & {FOB-I} & {NOB-I} & {FOB-I} & {NOB-I} & {FOB-I} & {NOB-I} & {FOB-II} & {NOB-II}\\
 \midrule

 500 & 90.3\% & 88.7\% & 75.4\% & 72.6\% & 13.9 & 13.9 & 20.4 & 20.5 &29.2 & 30.4 &9.2&10.1 \\
 & {$(91.3)$}& {$(90.1)$}&{$(81.7)$} &{$(79.5)$} & & & & & & & \\
1000 & 92.0\%& 90.3\%&82.5\% &82.4\% & 10.4 & 10.4 &12.9 & 12.9 & 21.2 & 22.0 &7.3 & 8.0\\
& {$(68.1)$}& {$(67.3)$}&{$(63.8)$} &{$(62.7)$} & & & & & & & & \\
2000 & 93.5\%& 91.7\%& 89.0\%& 85.6\%& 7.3&7.3 &8.0 &8.0 &14.3 &14.7 &5.1 &5.6 \\
& {$(49.6)$}& {$(49.6)$}& {$(47.5)$}& {$(46.9)$}& & & & & & & \\
5000 & 95.0\%& 93.7\%&93.0\% & 91.6\%& 4.5& 4.5&4.1 &4.1 &8.3 & 8.6& 3.1&3.3\\
& {$(32.4)$}& {$(32.9)$}& {$(31.9)$}& {$(32.0)$}& & & & & & & \\
\bottomrule
\end{tabularx}
\end{table}

The numbers in Table~\ref{tab:gamma-perf} are obtained from $1000$ macro-replications and assimilating their constructed confidence intervals and estimated statistical functionals $\hat{\psi}_n^\text{b}$, $\hat{\psi}_n^\sigma$ and $\hat{\psi}_n^\text{q}(\gamma=80)$ with OB-I and OB-II. We also test the impact of overlap by introducing two versions for each of OB-I and OB-II, namely, FOB-I and FOB-II that use fully-overlapping batches with $d_n=1$, and NOB-I and NOB-II that use non-overlapping batches with $d_n=m_n$. The steps below compute the estimators and their mean-squared-error quality for each of the three functionals, as well as the confidence intervals' coverage and half-width quality:

\begin{enumerate}
    \item approximate $\theta(P)$ with $\tilde \theta(P)$ by sending $n\to\infty$ (e.g., $n=10^6$).
    \item approximate $\psi_n$ for $n\in\{500,1000,2000,5000\}$ with a side experiment:
    \begin{enumerate}
        \item compute the estimator-error $\epsilon_{n}^{k'}=\theta(P_{n}^{k'})-\tilde \theta(P)$ for $k'=1,2,\ldots,10^5$;
        \item approximate $\psi_n$ with $\tilde \psi_n$ from $ \epsilon_{n}^{k'},\ k'=1,2,\ldots,10^5$.
    \end{enumerate}
    \item approximate $\text{RMSE}(\hat\psi_n,\psi_n)$:
    \begin{enumerate}
        \item compute $\hat \psi_{n}^k$ using the overlapping batches and their estimator-errors $ \epsilon_{i,n}^k$ for macroreplications $k=1,2,\ldots,1000$ following Table~\ref{table:psi};
        \item compute $\widehat{\text{RMSE}}(\hat\psi_n,\tilde\psi_n)=\left(\sum_{k=1}^{1000} \|\hat \psi_{n}^k - \tilde \psi_n\|_2^2/1000 \right)^{1/2}$.
    \end{enumerate}
    \item approximate the coverage probability and mean half-width of the confidence interval $C_n$:
    \begin{enumerate}
        \item compute $C_{n}^{k}$ for macroreplications $k=1,2,\ldots,1000$ following~\eqref{ob1cr} and~\eqref{ob2cr};
        \item compute $\hat{\mathbb P}(\theta(P)\in C_n)=\sum_{k=1}^{1000}\mathbb I(\tilde \theta(P)\in C_{n}^{k})/1000$ as the estimated coverage probability;
        \item compute $\hat r(C_n)=\sum_{k=1}^{1000}r(C_{n}^k)/1000$ as the estimated mean half-width  where 
        $r(C_{n}^k)=\text{len}(C_n^k)/2$ is the $k$-th confidence interval's approximated half-width.
    \end{enumerate}
\end{enumerate}

Next, see Figure~\ref{fig:gamma-p7} where we compare the histograms of the $P_{\mbox{\tiny OB-I},\varepsilon_n}$ and $P_{\mbox{\tiny OB-II},\varepsilon_n}$ (that use  $\sqrt{\frac{m_n}{n}}\varepsilon_{i,n}$, the scaled errors from batches--see Remark~\ref{rem:scale} and Table~\ref{table:psi}) with those of the sampling distribution $P_{\varepsilon_n}$ using $n$ observations. As depicted by the figure and report summary in the header of each plot (corresponding to each data size $n$), the bias, the standard deviation, and the 80\%-quantiles of $\varepsilon_n$ inferred from batch estimates for both OB-I and OB-II become increasingly accurate. We point out some key observations from this figure:
\begin{itemize}
    \item[-] With the choice of $m_n=\sqrt{n}$ and $d_n=1$, the batch size remains relatively small with $n$ while the number of batches $b_n$ grows fast helping with a better inference on the distribution of error. 
    \item[-] With smaller $n$ the distribution of error $P_{\varepsilon_n}$ is more skewed to the right.
    \item[-] $P_{\mbox{\tiny OB-I},\varepsilon_n}$ and $P_{\mbox{\tiny OB-II},\varepsilon_n}$ seem to somewhat capture the skew in $P_{\varepsilon_n}$, and as expected, the quality of approximation improves with $n$.
    \item[-] The bias of the error $\psi^{\text{b}}_n$ is always negative but decreases with $n$. The bias estimates $\hat{\psi}^{\text{b}}_{\mbox{\tiny OB-I},n}$ become closer to $\psi^{\text{b}}_n$ with $n$. OB-II does not provide a meaningful bias estimator as $\hat{\psi}^{\text{b}}_{\mbox{\tiny OB-II},n}=0$ always.
    \item[-] The standard deviation of the error $\psi^{\sigma}_n$ drops from $44$ to $14$ with $n$, becoming the major source of error compared to bias. OB-I and OB-II provide the same estimates $\hat{\psi}^{\sigma}_{\mbox{\tiny OB-I},n}=\hat{\psi}^{\sigma}_{\mbox{\tiny OB-II},n}$ that more accurately follow $\psi^{\sigma}_n$ with $n$.
    \item[-] The 0.8 quantile of the error $\psi^{\text{q}}_n$ drops from $27$ to $11$ with $n$, that is more slowly than its standard deviation. As both OB-I and OB-II generate more accurate estimates of the error's 0.8-quantile with $n$, we observe that $\hat{\psi}^{\text{q}}_{\mbox{\tiny OB-II},n}$ is always closer to $\psi^{\text{q}}_n$ than $\hat{\psi}^{\text{q}}_{\mbox{\tiny OB-I},n}$.
\end{itemize}
\begin{figure}[h]
\centering
    \includegraphics[width=\textwidth]{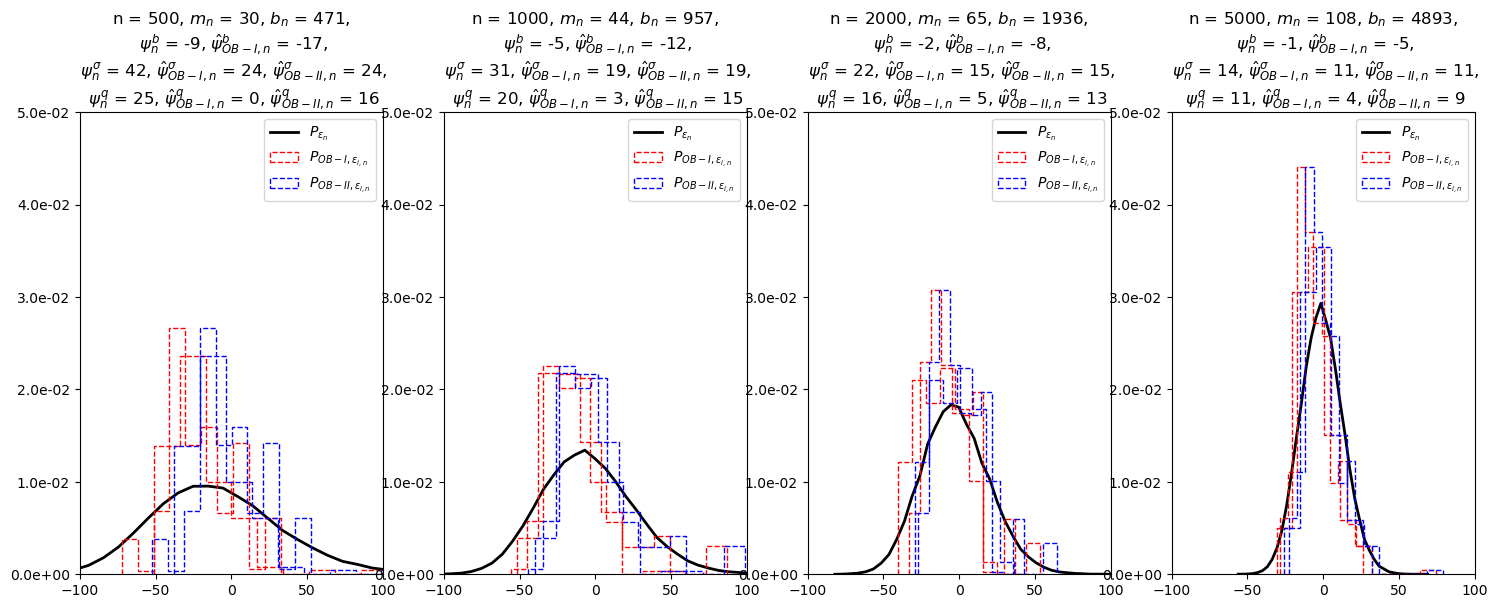}
    \caption{The error of estimating the 99\% quantile of a gamma distribution approximated with OB-I and OB-II and growing number of observations $n$ using $m_n=\sqrt{n}$ and $d_n=1$ (i.e., fully-overlapping batches).}
    \label{fig:gamma-p7}
\end{figure}

Lastly, to provide further perspective on the effect of the batch size $m_n$, number of batches $b_n$, and the offset $d_n$, on the coverage probability and expected half-width displayed in Table~\ref{tab:gamma-perf}, we systematically vary these parameters to generate the curves in Figure~\ref{fig:ob-nob-trends}. The curves in Figure~\ref{fig:ob-nob-trends} convey some clear trends. For instance, for the fully overlapping batch case, large batch sizes consistently yield better coverage although they are accompanied by a nearly linearly increasing expected half-width. The corresponding trend for the non-overlapping case is not as clear since as the batch size $m_n$ increases, the number of batches $b_n$ necessarily decreases. Furthermore, the curves in the non-overlapping context also suggest a generally higher variance than analogous curves for the overlapping context.

\begin{figure}[h]
\centering
    \includegraphics[width=0.7 \textwidth]{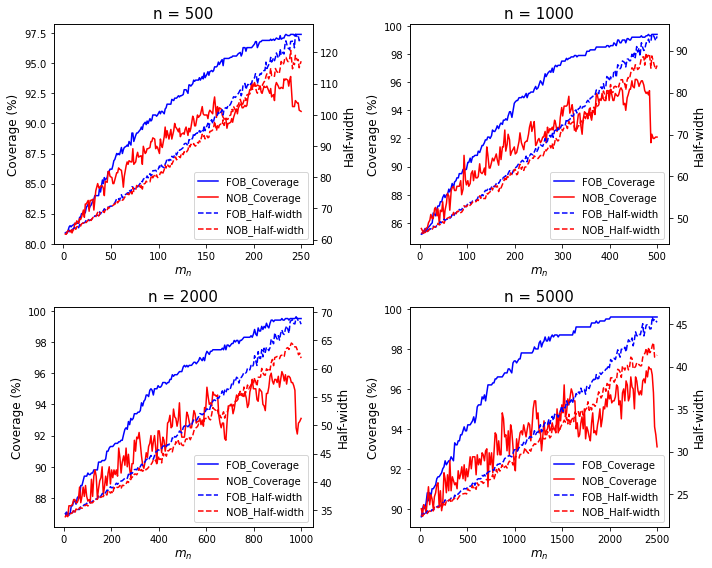}
    \caption{Trends in coverage (Left Y-axis) and half-widths (Right Y-axis) for fully overlapping (FOB) and non-overlapping (NOB) batches across varying batch size and observations with 1000 macro-replications.}
    \label{fig:ob-nob-trends}
\end{figure}

The next experiments were conducted in SimOpt \citep{SimOptsoftware2022}--an open-source testbed of SO problems and solvers, wherein solvers are coded to support flexible experimentation. All experiments use CRN by default. Our intention in using CRN across scenarios is to better estimate the \emph{ordering} of performance metrics when comparing multiple batching techniques. 

\subsection{Batching Inference with none-I.I.D. Observations: Total Cost of an Inventory System}

The inventory problem features a single product whose inventory level is dictated by an ($s$, $S$) inventory policy, meaning that when the inventory level drops below $s$, an order is placed so that the inventory level can be brought back up to $S$. The demand $D$ in each day has a mean of 100 units and a variance of 10000, independent across days. The order lead time $L$ is assumed to be discrete and random, following a Poisson distribution with mean of 6 days, independent across days. The objective is to estimate the expected daily total cost, which is the sum of back-order cost ($b=$\$$4$ per unit),  holding cost ($h=$\$$1$ per unit per day), fixed cost ($f=$\$$36$ per order), and variable cost ($v=$\$$2$ per unit). Let the daily inventory level based on the ($s$, $S$) inventory policy, the realized sequence of demands $\{D_j,\ j=1,2,\ldots,n\}$ and lead times $\{L_{j_k},\ k=1,2,\ldots\}$ for the $k$-th time of placing an order be $\{W_j,\ j=1,2,\ldots,n\}$. 
Then the daily total cost will be $\{Y_j,\ j=1,2,\ldots,n\}$ where \[Y_j=\underbrace{f\mathbb I(W_j\leq s)+v(S-W_j)\mathbb I(W_j\leq s)}_{\text{ordering cost (fixed and variable)}}\underbrace{+hW_j\mathbb I(W_j>0)}_{\text{holding cost}}\underbrace{-bW_j\mathbb I(W_j<0)}_{\text{back-order cost}}.\]

We use a warm-up of $1000$ days to ensure the stationarity condition outlined in Assumption~\ref{ass:stationarity} and change the number of days after warm-up from $n=500$ to $n=5000$ for a long-run simulation study. We consider $\theta=\min\{y:\Pr(Y\leq y)\geq 90\%\}$ (the steady-state $90\%$-quantile of the daily costs) and gamma-distributed demands and generate the same sampling distributions using batching as done in the i.i.d. random variate generation example. 

The main difference between the previous example and here is the existence of a simulation logic model and observations being time-dependent. The results of approximating the estimator error with batching is, however, consistent with the i.i.d. experiment. Figure~\ref{fig:sim-ex} illustrates the ability of error distributions generated with batching to capture particular properties of the sampling distribution in particular its variance and quantiles.

\begin{figure}[h]
\centering
    \includegraphics[width=\textwidth]{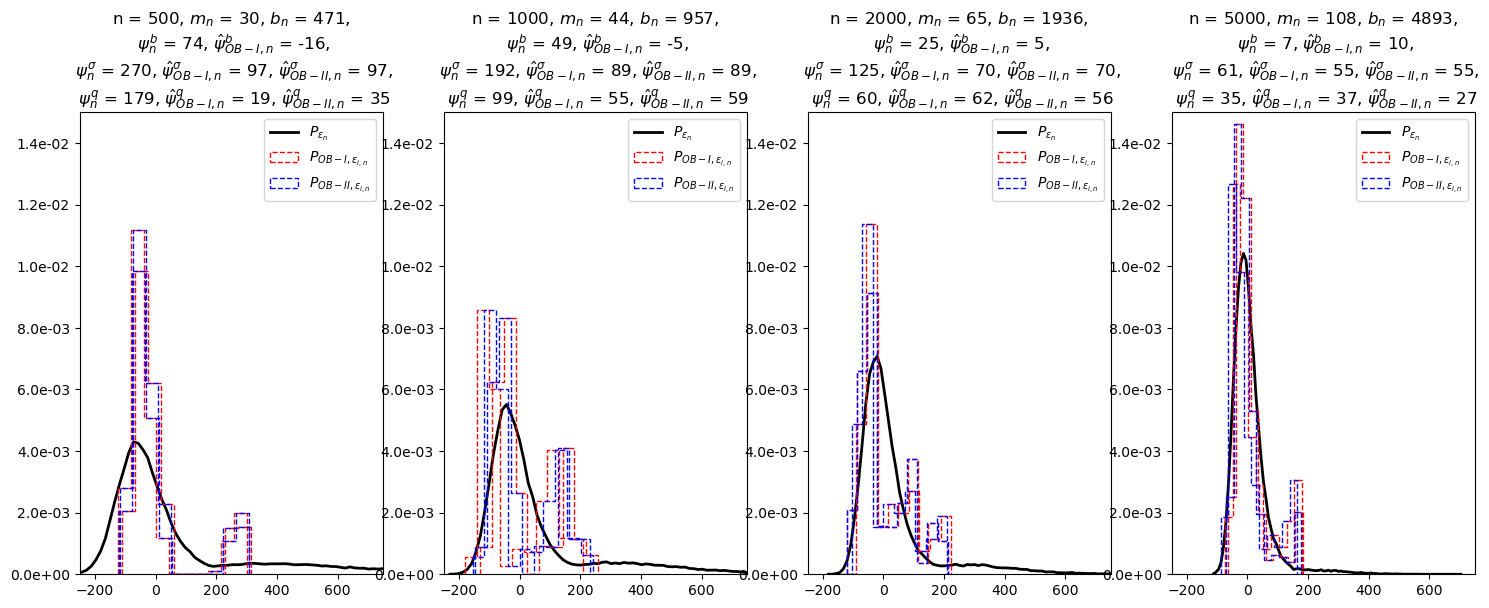}
    \caption{The error of estimating the 90\% quantile of a gamma distribution approximated with batching with growing number of observations $n$ and choice of $m_n=\sqrt{n}$ and $d_n=1$ (i.e., fully-overlapping batches).}
    \label{fig:sim-ex}
\end{figure}

A general practitioner would use simulation outputs from $n$ days to make an inference about the inventory system. Their judgement on the goodness of their estimates with the batching schemes proposed in the paper would provide them with a confidence interval and an estimate of bias, variance, or quantiles of the error. Table~\ref{tab:single-exp-sim} provides these inferences that the experimenter would obtain with the approximate true value of $\theta(P)\approx 1848.8$ and the true $\psi_n$ values listed at the bottom for comparison.

\begin{table}[h]  
\caption{Single experiment inference on inventory costs with gamma-distributed demand for 500 days. The FOB-I and FOB-II use full overlap. The NOB-I and NOB-II use no overlap. FOB-I has the closest estimation for 0.9-quantile but with $n=500$ all cases underestimate the bias and standard deviation.}
\label{tab:single-exp-sim}
\centering
\small
\footnotesize
\begin{tabularx}{\linewidth}{X *{5}{S[table-format=2.1]}}
\toprule
& \text{95\% CI ($C_n$)} & \text{bias ($\hat{\psi}^{\text{b}}_n$)}  &\text{standard deviation ($\hat{\psi}^{\sigma}_n$)}  & \text{0.8-quantile ($\hat{\psi}^{\text{q}}_n$)} \\
\midrule
FOB-I := OB-I($d_n=1$) &  \text{[936.6, 2421.7]} & {28.2} & 138.2 & 169.5 \\
\addlinespace[0.5em]
NOB-I := OB-I($d_n=m_n$)  & \text{[952.0, 2406.4]} & {22.4} &  137.9 & 131.5 \\
\addlinespace[0.5em]
FOB-II := OB-II($d_n=1$) & \text{[1326.7, 2657.0]} & {-} & 138.2 &  141.3\\
\addlinespace[0.5em]
NOB-II := OB-II($d_n=m_n$)  & \text{[1281.4, 2619.8]} & {-} & 137.9 &  109.1\\
\midrule
$\psi_{n}$ & {-}&  {74.7}& {272.4} & {177.2} \\       
\bottomrule
    \end{tabularx}
\end{table}

We next investigate the performance both in terms of accurate confidence intervals and $\psi$ estimators following the same procedure as in the previous experiment. Table~\ref{tab:cost-perf} summarizes the findings. Note, again that OB-II method does not estimate the bias and has the same standard deviation estimate as OB-I. The best coverage is achieved by FOB-I in every $n$ although at the expense of slightly elongated half-widths. Simultaneously, we also observe a rapid drop in the error of estimating bias with $n$. Error of estimating standard deviation always dominates that of bias in magnitude but it also rapidly drops with $n$. In this simulation study the amount of overlap does not make a significant difference in terms of estimates of bias or standard deviation. For the 0.8-quantiles, we see some differences. FOB-II achieves the lowest error except in the case of $n=5000$. It also suggests that full overlap is effective with the non-i.i.d. observations. It is possible that the optimal choice of batch size may vary with $n$ but we leave that as an open question.

\begin{table}[h] 
\caption{Estimation quality of total cost 0.90-quantile in the inventory system with gamma-distributed demand using 1000 macro-replications. NOB-I and NOB-II stand for non-overlapping batches with OB-I($d_n=m_n$) and OB-II($d_n=m_n$).}
\label{tab:cost-perf}
\small
\footnotesize
\scriptsize 
\tiny 
\centering
\begin{tabularx}{\linewidth}{X *{13}{S[table-format=2.1]}}
\toprule
& \multicolumn{4}{c}{CI with $m_{n} = 0.2n$} &  \multicolumn{8}{c}{RMSE of $\hat{\psi}_n$ with $m_{n} = \sqrt{n}$} \\
\cmidrule(lr){2-5}
\cmidrule(lr){6-13}
 & \multicolumn{4}{c}{coverage and half-width} & \multicolumn{2}{c}{bias} & \multicolumn{2}{c}{standard deviation} & \multicolumn{4}{c}{0.8-quantile} \\
 \cmidrule(lr){2-5}
 \cmidrule(lr){6-7}
 \cmidrule(lr){8-9}
 \cmidrule(lr){10-13}
 $n$  & {FOB-I} & {NOB-I} & {FOB-II} & {NOB-II} & {FOB-I} & {NOB-I} & {FOB-I} & {NOB-I} & {FOB-I} & {NOB-I} & {FOB-II} & {NOB-II}\\
 \midrule
500 & 93.3\% & 88.5\% & 92.4\% & 86.8\% & 115.8& 116.9& 164.1& 166.5 &158.4 & 167.4&130.3 & 134.0\\
 & {$(415.1)$}& {$(400.0)$}&{$(391.5)$} & {$(370.1)$} & & & & & & & & & \\
\addlinespace[0.5em]
1000 & 92.2\% & 89.4\% & 91.9\% & 88.6\% & 74.9 & 75.4 & 107.5 & 109.2 &93.9 & 97.1&75.3 & 76.8\\
& {$(363.9)$}& {$(346.0)$}&{$(346.6)$} & {$(330.0)$}& & & & & & & & \\
\addlinespace[0.5em]
2000& 93.9\%& 91.6\%&93.4\% &90.5\% & 28.2& 28.7&59.1 & 59.8 &27.0 &32.5 &25.9 &29.3 \\
& {$(289.3)$}& {$(276.7)$}&{$(278.9)$} &{$(266.1)$} & & & & & & & & \\
\addlinespace[0.5em]
5000& 94.7\%& 91.9\%&92.8\% &91.1\% & 5.1& 5.6&14.1 & 14.8 &17.9 & 19.7&21.2 & 22.3\\
& {$(179.1)$}& {$(176.1)$}&{$(174.3)$} &{$(171.1)$} & & & & & & & & \\
\bottomrule
\end{tabularx}
\end{table}

\begin{figure}[h]
\centering
    \includegraphics[width=\textwidth]{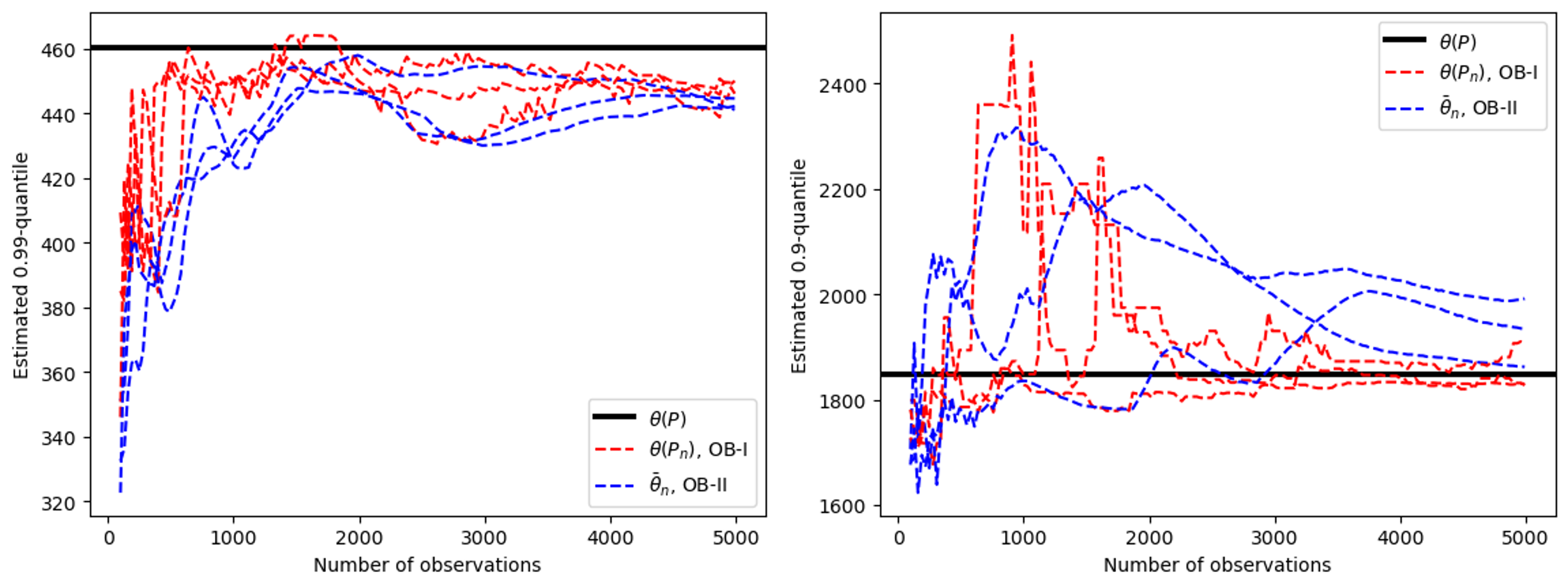}
    \caption{Quantile estimation trends between OB-I and OB-II under two different scenarios (left: i.i.d., right: non-i.i.d.) with three independent replications.} 
    \label{fig:comparison}
\end{figure}

Lastly, we conclude the experimental study by highlighting the differences between OB-I and OB-II using the behavior of point estimates as a function of $n$. Figure~\ref{fig:comparison} illustrates that compared to OB-I, OB-II tends to have less fluctuation. In the non-i.i.d. experiment, the subsequent (nearby) observation is more likely to be smaller when the peak-point occurs (i.e., due to the penalized backordering cost), leading to momentary impact on $\theta_{n}$ causing volatile OB-I estimate (right panel of Figure~\ref{fig:comparison}).  Larger $n$ values stabilize OB-I more quickly. OB-II estimates while smoother, take a longer time to stabilize. In the i.i.d. experiment (left panel of Figure~\ref{fig:comparison}), OB-I and OB-II have a similar trend of exhibiting negative bias.
Depending on the properties of the simulation output, the number of batches, size of the batch, and lastly the amount of overlap (offset), the accuracy of OB-I and OB-II point estimates can be different. We leave this for a future investigation.

\section{SUMMARY OF INSIGHT AND POSTSCRIPT}

We summarize several insights that have emerged through our analysis of batching estimators for simulation output analysis.  

\begin{enumerate} 
 
\item[(a)] Much like bootstrapping, batching can be reliably used as an omnibus device for uncertainty quantification using simulation output.  

\item[(b)] The appropriate choice of batch sizes and the number of batches are crucial to ensure consistency and higher-order accuracy. Broadly, large overlapping batches perform best when constructing confidence regions in the sense of ensuring fastest convergence to nominal coverage probability. Large batches cannot be used when estimating $\psi$ due to the resulting loss in consistency. 

\item[(c)] A frequent question among simulation practitioners is whether resampling is needed if ``additional simulation runs can be performed'' easily. This question becomes moot if efficiency in the sense of teasing out more information from a \emph{given} amount of data is of interest. Batching and bootstrapping are methods that allow for efficient inference.  

\item[(d)] Our treatment assumes that the simulation output data $(Y_1,Y_2, \ldots,Y_n)$ are in steady state. This is usually not the case in practice, leading to what has been called the \emph{initial transient problem}. See~\cite{2010passch} for an annotated bibliography on this problem. For appropriate inference, ideas from removing the initial transient need to be used in concert with batching, constituting what is an interesting research question. \item[(e)] A consistent estimator of the variance constant $\Sigma$ is neither needed nor preferred when constructing a confidence region using batching. Interestingly, however, in the sequential context where the data $Y_1,Y_2, \ldots$ are revealed one by one, a risk-optimal estimator of $\theta$ might entail consistently estimating $\sigma^2$. See~\cite{2020pasyeh} for more. \item[(f)] Virtually all  discussion in this paper applies to estimators constructed in the context of \emph{digital twins}~\citep{2022biletal}.
\item[(g)] Parametric batching, analogous to parametric bootstrap~\citep{2017che}, has not been sufficiently explored and should form a topic of future research.
\item[(h)] Bias estimation tends to be tricky and delicate, and should be performed with care. This issue is not specific to batching and similar caution has been issued even in the context of the bootstrap and the jackknife~\citep{1994efrtib}. 
\item[(i)] There is a deep and interesting connection between variance estimation and certain types of input model uncertainty, as explained through semi-parametric estimation~\citep{2008kos}. The treatment here has relevance to input model uncertainty when seen through the semi-parametric estimation lens.

\item[(j)] OB-I and OB-II methods presented here need to be formally compared against bootstrapping and against each other. Numerical experience suggests that OB-I and OB-II outperform bootstrapping in contexts involving dependent data but a theoretical statement confirming such dominance is an open question.

\end{enumerate}

\section*{Acknowledgements}  Sara Shashaani gratefully acknowledges the U.S. National Science Foundation for support provided by grant number CMMI-2226347. 
Raghu Pasupathy gratefully acknowledges the U.S. Office of Naval Research for support provided through grants N000141712295 and 13000991. He also thanks Prof. Bruce Schmeiser and Prof. Peter Glynn for many insightful discussions on the topic of inference. 

\bibliographystyle{apalike}
\bibliography{stochastic_optimization}

\end{document}